\documentclass[a4paper,UKenglish,cleveref, autoref, thm-restate]{lipics-v2021} \hideLIPIcs \nolinenumbers
\pdfoutput=1
\usepackage{amsmath}
\usepackage{bm}
\usepackage{braket}
\usepackage{mdframed}
\newcommand{\congest}{{\mathsf{CONGEST}}}
\newcommand{\clique}{{\mathsf{CONGEST~CLIQUE}}}
\newcommand{\local}{{\mathsf{LOCAL}}}

\title{Distributed Quantum Interactive Proofs} 

\author{Fran\c{c}ois Le Gall}{Graduate School of Mathematics, Nagoya University, Nagoya, Japan}{legall@math.nagoya-u.ac.jp}{}{}
\author{Masayuki Miyamoto}{Graduate School of Mathematics, Nagoya University, Nagoya, Japan }{masayuki.miyamoto95@gmail.com}{}{}
\author{Harumichi Nishimura}{Graduate School of Informatics, Nagoya University, Nagoya, Japan} {hnishimura@is.nagoya-u.ac.jp}{}{}

\authorrunning{F. Le Gall, M. Miyamoto and H. Nishimura} 
\Copyright{Fran\c{c}ois Le Gall, Masayuki Miyamoto and Harumichi Nishimura}
\category{}
\relatedversion{}
\funding{FLG was supported by the JSPS KAKENHI grants JP16H01705, JP19H04066, JP20H00579, JP20H04139, JP20H05966, JP21H04879 and by the MEXT Q-LEAP grants JPMXS0118067394 and JPMXS0120319794. MM was supported by JST, the establishment of University fellowships towards the creation of science technology innovation, Grant Number JPMJFS2120. HN was supported by the JSPS KAKENHI grants JP19H04066, JP20H05966, 
JP21H04879, JP22H00522 and by the MEXT Q-LEAP grants JPMXS0120319794.}

\EventEditors{John Q. Open and Joan R. Access}
\EventNoEds{2}
\EventLongTitle{42nd Conference on Very Important Topics (CVIT 2016)}
\EventShortTitle{CVIT 2016}
\EventAcronym{CVIT}
\EventYear{2016}
\EventDate{December 24--27, 2016}
\EventLocation{Little Whinging, United Kingdom}
\EventLogo{}
\SeriesVolume{42}
\ArticleNo{23}
\begin{document}
\maketitle
\begin{abstract}
The study of distributed interactive proofs was initiated by Kol, Oshman, and Saxena [PODC 2018] as a generalization of distributed decision mechanisms (proof-labeling schemes, etc.), and has received a lot of attention in recent years. In distributed interactive proofs, the nodes of an $n$-node network $G$ can exchange short messages (called certificates) with a powerful prover. The goal is to decide if the input (including $G$ itself) belongs to some language, with as few turns of interaction and as few bits exchanged between nodes and the prover as possible. There are several results showing that the size of certificates can be reduced drastically with a constant number of interactions compared to non-interactive distributed proofs.

In this paper, we introduce the quantum counterpart of distributed interactive proofs: certificates can now be quantum bits, and the nodes of the network can perform quantum computation. The first result of this paper shows that by using quantum distributed interactive proofs, the number of interactions can be significantly reduced. More precisely, our result shows that for any constant~$k$, the class of languages that can be decided by a $k$-turn classical (i.e., non-quantum) distributed interactive protocol with $f(n)$-bit certificate size is contained in the class of languages that can be decided by a $5$-turn distributed quantum interactive protocol with $O(f(n))$-bit certificate size.  We also show that if we allow to use shared randomness, the number of turns can be reduced to 3-turn. Since no similar turn-reduction \emph{classical} technique is currently known, our result gives evidence of the power of quantum computation in the setting of distributed interactive proofs as well. 

As a corollary of our results,
we show that there exist 5-turn/3-turn distributed quantum interactive protocols with small certificate size for problems that have been considered in prior works on distributed interactive proofs such as 
[Kol, Oshman, and Saxena PODC 2018, Naor, Parter, and Yogev SODA 2020].

We then utilize the framework of the distributed quantum interactive proofs to test closeness of two quantum states each of which is distributed over the entire network.

\end{abstract}

\maketitle

\newpage

\section{Introduction}
\subsection{Distributed Interactive Proofs}
In distributed computing, efficient verification of graph properties of the network is useful from both theoretical and applied aspects. The study of this notion of verification in the distributed setting has lead to the notion of "distributed~$\mathsf{NP}$" in analogy with the complexity class $\mathsf{NP}$ in centralized computation: %\cite{korman2010proof,fraigniaud2011local,goos2016locally}. The distributed version is inspired by the concept of $\mathsf{NP}$ in the centralized model: 
A powerful prover provides certificates to each node of the network in order to convince that the network has a desired property; 
%the property we want to verify. 
%(or more generaly, input labels given to each node satisfies some boolean predicate). 
If the property is satisfied, all nodes must output "accept", otherwise at least one node must output "reject". This concept of "distributed~$\mathsf{NP}$" has been formulated in several ways, including \textit{proof-labeling schemes} (PLS)~\cite{korman2010proof}, \textit{non-deterministic local decision} (NLD)~\cite{fraigniaud2011local}, and \textit{locally checkable proofs} (LCP)~\cite{goos2016locally}.

As a motivating example, consider the problem of verifying whether the network is bipartite or not. While this problem cannot be solved in $O(1)$ round without prover \cite{sarma2012distributed}, it can easily be solved with a prover telling to each node to each part it belongs to, which requires only a $1$-bit certificate per node, and then each node broadcasting this information to its adjacent nodes (here the crucial point is that if the network is non-bipartite, then at least one node will be able to detect it). On the other hand, it is known that there exist  properties that require large certificate size to decide: G{\"o}{\"o}s and Suomela \cite{goos2016locally} have shown that recognizing symmetric graphs (\textsc{Sym}) and non 3-colorable graphs ($\overline{\textsc{3Col}}$) require $\Omega(n^2)$-bit certificates per node in the framework of LCP (which is tight since all graph properties are locally decidable by giving the $O(n^2)$-bit adjacency matrix of the graph).

To reduce the length of the certificate for such problems, the notion of distributed interactive proofs (also called distributed Arthur-Merlin proofs) was recently introduced by Kol, Oshman and Saxena \cite{kol2018interactive} as a generalization of distributed $\mathsf{NP}$. In this model there are two players, the prover (often called Merlin), who has unlimited computational power and sees the entire network but is untrusted (i.e., can be malicious), and the verifier (often called Arthur) representing all the nodes of the network, who can perform only local computation and brief communication with adjacent nodes. Generalizing the concept of distributed $\mathsf{NP}$, the nodes are now allowed to engage in multiple turns of interaction with the prover.
%The prover Merlin has unlimited computational power and sees the entire network (but may be malicious), and the distributed verifier Arthur only knows its local input at each node. 
%The power of interaction emerges when the nodes behave randomly (interaction is not useful for deterministic verifiers), 
As for distributed $\mathsf{NP}$, there are two requirements of the protocol: if the input is legal (yes-instance) then all nodes must accept with high probability (\textit{completeness}), and if the input is illegal then at least one node must reject with high probability (\textit{soundness}). 

In the setting of \cite{kol2018interactive}, each node has access to a private source of randomness, and sends generated random bits to the prover in Arthur's turn. For instance, a 2-turn protocol contains two interactions: Arthur first queries Merlin by sending a random string from each node, and then Merlin provides a certificate to each node. After that, nodes exchange messages with adjacent nodes to decide their outputs. The main  complexity measures when studying distributed interactive protocols
%, which we mention as $\mathsf{dAM}$ complexity, 
are the size of certificates provided to each node, the size of the random strings generated at each node and the size of the messages exchanged between nodes. Let us denote $\mathsf{dAM}[k](f(n))$ the class of languages that have $k$-turn distributed Arthur-Merlin protocols where Merlin provides $O(f(n))$-bit certificates, Arthur generates $O(f(n))$-bit random strings at each node and $O(f(n))$-bit messages are exchanged between nodes. Kol et al.~\cite{kol2018interactive} showed the power of interaction by giving a $\mathsf{dMAM}(\log n) = \mathsf{dAM}[3](\log n)$ protocol for graph symmetry (\textsc{Sym}) and a $\mathsf{dAMAM}(n\log n) = \mathsf{dAM}[4](n\log n)$ protocol for graph non-isomorphism (\textsc{GNI}), which are known to require $\Omega(n^2)$-bit certificate in LCP (see Appendix~\ref{appendix:problems} for the definition of these problems).

This model has been further studied in several works. Naor, Parter and Yogev \cite{naor2020power} showed that any $O(n)$-time centralized computation can be converted
into a $\mathsf{dMAM}(\log n)=\mathsf{dAM}[3](\log n)$ protocol. Using this compiler, for instance, %they constructed a $\mathsf{dAMAM}(\log n)=\mathsf{dAM}[4](\log n)$ protocol for $\textsc{GNI}$, and 
they constructed a $\mathsf{dMAMAM}(\log \log n)=\mathsf{dAM}[5](\log \log n)$ protocol for $\textsc{SetEquality}$ and a special case of $\textsc{Sym}$. Crescenzi, Fraigniaud and Paz \cite{crescenzi2019trade} initiated the study of distributed Arthur-Merlin protocols with shared randomness: in each Arthur's turn, Arthur generates a random string that can be seen from all nodes. In order to distinguish the two models we use $\mathsf{dAM}$ for the (standard) private randomness setting and $\mathsf{dAM}^{sh}$ for the shared randomness setting. %(if it is clear which randomness we are concerning about, we simply denote $\mathsf{dAM}$.). 
They showed that $\mathsf{dAM}$ protocols can simulate $\mathsf{dAM}^{sh}$ protocols by giving additional $O(\log n)$-size certificates. The role of shared randomness was further investigated by Montealegre, Ram{\'\i}rez-Romero and Rapaport \cite{montealegre2020shared}, who showed the computational power of small-certificate $\mathsf{dAM}^{sh}$ protocols without private randomness is relatively weak: for any constant $k$, $\mathsf{dAM}^{sh}[k]$ protocols with message size $m$ can be converted to locally checkable proofs (LCPs) with message size $O(2^m + \log n)$. %which does not hold for the standard $\mathsf{dAM}$ protocols (as mentioned above, the language $\textsc{DSym}$, which is a special case of $\textsc{Sym}$, cannot be decided by LCPs with message size $o(n^2)$ while there exists a $\mathsf{dMAMAM}(\log \log n)=\mathsf{dAM}[5](\log \log n)$ protocol for $\textsc{DSym}$). They also showed that in the case of $\dMA$ protocols (randomized proof-labeling schemes), shared randomness is more powerful than private randomness.  

Lower bounds on distributed Arthur-Merlin protocols for some concrete problems are known. Kol, Oshman and Saxena \cite{kol2018interactive} showed that if the language $\textsc{Sym}$ is in the class $\mathsf{dAM}[2](f(n))$, then $f(n)\in \Omega(\log \log n)$. As mentioned in \cite{fraigniaud2019distributed}, this lower bound can actually be improved to $f(n)\in \Omega(\log n)$. On the other hand, there is no known method to prove lower bounds when the number of turns is three or more.

\subsection{Quantum Interactive Proofs}\label{subsec:qip}
%Remember that in the centralized classical (i.e., non-quantum) setting interactive proof ($\mathsf{IP}$) systems were first introduced as a generalization of $\mathsf{NP}$ (and $\mathsf{MA}$) \cite{babai1985trading,goldwasser1989knowledge}: In this system a verifier is a $\mathsf{BPP}$ machine, and can interact with a powerful but untrusted prover to solve some problem. It is known that the class of languages that can be efficiently decided by $\mathsf{IP}$s with a polynomial number of interaction 
%of polynomial size
%is equal to the class $\mathsf{PSPACE}$ \cite{lund1992algebraic,shamir1992ip} of languages that can be decided with polynomial space. 
Quantum interactive proofs ($\mathsf{QIP}$) were introduced by Watrous \cite{watrous2003pspace} in the centralized setting as a variant of classical interactive proofs ($\mathsf{IP}$) in which the verifier can perform polynomial-time quantum computation (instead of polynomial-time classical computation), and the prover and verifier can exchange quantum bits (instead of classical bits). Kitaev and Watrous \cite{kitaev2000parallelization} first showed that $\mathsf{QIP}$, the class of languages that can be decided by a quantum interactive protocol with polynomial number of interactions, is contained in $\mathsf{EXP}$, the class of languages decided in exponential time.  This containment was improved by Jain, Ji, Upadhyay, and Watrous \cite{jain2011qip}, who showed that $\mathsf{QIP}$ is actually contained in $\mathsf{PSPACE}$, which implies that $\mathsf{QIP}$ collapses to the complexity class $\mathsf{IP}$  ($\mathsf{QIP}=\mathsf{IP}=\mathsf{PSPACE}$). 

While the above result shows that quantum interactive proofs are not more powerful than classical interactive proofs, there is a striking property of quantum interactive proofs that is not expected to hold for classical interactive proofs: in the quantum case the number of interactions can be significantly reduced. More precisely, Watrous first showed that any language in $\mathsf{PSPACE}$ can be decided by a three-turn $\mathsf{QIP}$ protocol \cite{watrous2003pspace}. After that, Kitaev and Watrous \cite{kitaev2000parallelization} showed that any $\mathsf{QIP}$ protocol with a polynomial number of interaction can be parallelized to three turns ($\mathsf{QIP} = \mathsf{QIP}[3]$). Marriott and Watrous \cite{marriott2005quantum} additionally showed that the verifier's turn in $\mathsf{QIP}[3]$ protocols can be replaced by a 1-bit coin flip ($\mathsf{QIP}[3]$=$\mathsf{QMAM}$). Kempe, Kobayashi, Matsumoto, and Vidick~\cite{kempe2009using} showed an alternative proof of $\mathsf{QIP} = \mathsf{QIP}[3]$.

\subsection{Our Results}\label{subsec:our_results}
In this paper we introduce the quantum counterpart of distributed interactive proofs, which we call distributed quantum interactive proofs (or sometimes distributed quantum interactive protocols) and write $\mathsf{dQIP}$, and show their power. 
Roughly speaking, distributed quantum interactive proofs are defined similarly to the classical distributed interactive proofs (i.e., distributed Arthur-Merlin proofs) defined above, but the messages exchanged between the prover and the nodes of the network can now contain quantum bits (qubits), the nodes can now do any (local) quantum computation (i.e., each node can apply any unitary transform to the registers it holds), and each node can now send messages consisting of qubits to its adjacent nodes. In analogy to the classical case, the main complexity measures when studying distributed quantum interactive protocols are the size of registers exchanged between the prover and the nodes, and the size of messages exchanged between the nodes. We give the formal definition of $\mathsf{dQIP}$ in Section \ref{section:definitions}. The class $\mathsf{dQIP}[k](f(n))$ is defined as the set of all languages that can be decided by a $k$-turn $\mathsf{dQIP}$ protocol where both the size of the messages exchanged between the prover and the nodes, and the size of the messages exchanged between the nodes are $O(f(n))$ qubits.

Our first result is the following theorem.

\begin{theorem}\label{theorem:main_theorem}%[Main Result]
For any constant $k\geq 5$, $\mathsf{dAM}[k](f(n))\subseteq \mathsf{dQIP}[5](f(n))$.
\end{theorem}

Theorem \ref{theorem:main_theorem} shows that by using distributed quantum interactive proofs, the number of interactions in distributed interactive proofs can be significantly reduced. To prove this result, we develop a generic \emph{quantum} technique for turn reduction in distributed interactive proofs. Since no similar turn-reduction \emph{classical} technique is currently known, our result gives evidence of the power of quantum computation in the setting of distributed interactive proofs as well. 

We also show that if we allow to use randomness shared to all nodes (we denote this model by $\mathsf{dQIP}^{sh}$), the number of turns can be further reduced to three turns. 
\begin{theorem}\label{theorem:main_theorem_shared}
For any constant $k\geq 3$, $\mathsf{dAM}[k](f(n))\subseteq \mathsf{dQIP}^{sh}[3](f(n))$.
\end{theorem}
On the other hand, in the classical case, it is known that allowing shared randomness does not change the class \cite{crescenzi2019trade}: $\mathsf{dAM}^{sh}[k](f(n))\subseteq \mathsf{dAM}[k](f(n))$ for all $k\geq 3$.\footnote{In fact, the authors of \cite{crescenzi2019trade} showed $\mathsf{dAM}^{sh}[k](f(n))\subseteq \mathsf{dAM}[k](f(n)+\log n)$ for all $k\geq 1$ where the additional $\log n$ comes from constructing a spanning tree, but for $k\geq 3$, a spanning tree can be constructed with $O(1)$-sized messages between the prover and the nodes in three turns~\cite{naor2020power}, so $\log n$ can be removed.}

As mentioned above, for (classical) $\mathsf{dAM}$ protocols increasing the number of turns %to more than three 
is helpful to reduce the complexity (in particular, the certificate size) for many problems. Our result thus shows if we allow quantum resource, such protocols can be simulated in five turns, and in three turns if we allow shared randomness. More precisely, we obtain the following corollary (see Appendix~\ref{appendix:problems} for the precise definitions of these problems and Theorems \ref{th:cl1} and \ref{th:cl2} in Section \ref{section:simulation} for a statement of the corresponding classical results):

\begin{corollary}\label{corollary:concrete_problems}
\begin{enumerate}
\item There exist 
    \begin{itemize}
        \item a $\mathsf{dQIP}^{sh}[3](\log n)$ protocol for $\textsc{Asym}$,
        \item a $\mathsf{dQIP}^{sh}[3](\log n)$ protocol for $\textsc{GNI}$,% where $\textsc{GNI}$ is the harder version defined in Section 2,
        \item a $\mathsf{dQIP}^{sh}[3](\log \log n)$ protocol for $\textsc{SetEquality}$,
        \item a $\mathsf{dQIP}^{sh}[3](\log \log n)$ protocol for $\textsc{DSym}$.
        \item a $\mathsf{dQIP}[5](\log n)$ protocol for $\textsc{GNI}$.
    \end{itemize} 
\item There exists a constant $\delta$ such that if
a language $\mathcal{L}$ can be decided in $\mathrm{poly}(n)$ time and $n^{\delta}$ space, then $\mathcal{L}\in \mathsf{dQIP}[5](\log n)$ and $\mathcal{L}\in \mathsf{dQIP}^{sh}[3](\log n)$.
\end{enumerate}
\end{corollary}

We also introduce a \textit{quantum} problem (i.e., the inputs are quantum states) which arises naturally when considering distributed quantum networks. More specifically, we consider the following task: There are two quantum states $\ket{\psi}$ and $\ket{\phi}$ as the inputs, each of which is distributed over the entire network (each node $u\in V$ has $N_u$-qubit of $\ket{\psi}$ and $\ket{\phi}$, where $\sum_{u\in V}N_u = N$). The goal of the task is to measure closeness of these states. We call this problem $N$-qubit Distributed Quantum Closeness Testing ($\mathsf{DQCT}_N$).
\begin{comment}
Here we formalize this task as $\mathsf{DQCT}_N$.
\begin{definition}[Distributed Quantum Closeness Testing]\label{def:dqct}
The network has two $N$-qubit quantum states $\ket{\psi}$ and $\ket{\phi}$ as the input in the following distributed manner: each node $u\in V$ has $N_u$-qubit of $\ket{\psi}$ and $\ket{\phi}$, where $\sum_{u\in V}N_u = N$. The aim is to
accept if the distance $\mathrm{dist}(\ket{\psi},\ket{\phi})$ (see Section~\ref{sec:quantum_information} for the definition of $\mathrm{dist}$) is small, and reject otherwise. We call this problem $N$-qubit Distributed Quantum Closeness Testing ($\mathsf{DQCT}_N$).
\end{definition}
\end{comment}
For this task, we show the following theorem.
\begin{theorem}\label{thm:closeness}
There is a $\mathsf{dQIP}[5](O(1))$ protocol for $\mathsf{DQCT}_N$, where the completeness and the soundness conditions are defined as follows:
\begin{itemize}
    \item \textbf{Completeness:} If $\ket{\psi} = \ket{\phi}$ and the prover is honest, the protocol is accepted with probability~1.
    %\item  For any $\ket{\psi}$ and $\ket{\phi}$, the protocol is accepted with probability at most $\frac{1}{2}+\frac{1}{2}|\braket{\psi|\phi}|^2 + \varepsilon$ for any small constant $\varepsilon > 0$.
    \item \textbf{Soundness:} If the protocol is accepted with probability $1-1/z$, $\mathrm{dist}(\ket{\psi},\ket{\phi}) \leq \sqrt{2/z} + \varepsilon$ for any small constant $\varepsilon > 0$.
\end{itemize}
\end{theorem}
Without the prover, a naive approach is to accumulate all of the input to the leader node, and perform local operations at the leader node to measure their closeness. Obviously this approach is inefficient in the following sense: (1) it requires $\Omega(D)$-round of communication where $D$ is the diameter of the network; (2) the amount of communication is linear in $N$, the size of the input quantum states. Theorem~\ref{thm:closeness} means that in the $\mathsf{dQIP}$ setting, (1) it only needs 1-round of communication between the nodes; (2) the amount of communication (the size of messages per edge, and the size of messages between each node and the prover) is $O(1)$, regardless of the input size. Note that the main result of the recent paper \cite{fraigniaud2021distributed} (see Section~\ref{subsec:related_work} for their result) immediately shows that if the two input quantum states are hold by some specific two nodes respectively and there is no input for the other nodes, $\mathsf{DQCT}_N$ in an $n$-node network can be done with $O(N\cdot \mathrm{poly}(n))$ size of quantum proofs and 1-round of communication between nodes, in the non-interactive setting. Our setting is more general in the sense that the input quantum states can be distributed over the entire network.

Lastly, in Appendix \ref{section:perfect}, we show how to transform $\mathsf{dQIP}$ protocols with two-sided (i.e., completeness and soundness) bounded error into $\mathsf{dQIP}$ protocols with perfect completeness. 
We show that if we allow the communication between nodes in the middle of interaction (we call this model as $\mathsf{dQIP}c$), achieving perfect completeness is possible. Thus $\mathsf{dQIP}c$ protocols can be converted to 5-turn protocols with perfect completeness using parallel repetition with a fairly small increase of the message size.

%\subsubsection{Brief overview of our approach.}  In a nutshell, the basic strategy to prove Theorem~\ref{theorem:main_theorem} is to convert the classical $k$-turn distributed Arthur-Merlin protocol into a quantum $k$-turn distributed quantum interactive protocol, and then apply our novel turn-reduction quantum technique. While this novel turn-reduction technique is inspired by the turn-reduction techniques by Kitaev and Watrous \cite{kitaev2000parallelization} for the centralized setting mentioned above, making this technique work in the distributed setting requires several new ideas and insights. Especially, the final step of the reduction (the conversion from a 5-round quantum protocol into a 3-round quantum protocol) requires special conditions on the 5-round protocol in order to be applicable. One key insight is that when starting from a classical $k$-turn distributed Arthur-Merlin protocol with $k$ constant, these conditions are satisfied and the conversion is possible.

\subsection{Organization and Overview of our Approach}
We start by considering a more powerful model than $\mathsf{dQIP}$, which allows nodes to use a shared randomness. That is, at each turn the network can send a shared random string of limited length to the prover. We call this model $\mathsf{dQIP}^{sh}$. In Section \ref{subsection:shared}, we first show how to reduce the number of turns by half in the $\mathsf{dQIP}^{sh}$ model. This is shown by adapting to the distributed setting the method of Kempe et al.~\cite{kempe2009using}, which reduces the number of turns of $\mathsf{QIP}$ by half. More precisely, we show that for any $\ell \geq 1$, $(4\ell + 1)$-turn $\mathsf{dQIP}^{sh}$ protocols can be parallelized to $(2\ell + 1)$-turn. The main idea of \cite{kempe2009using} is the following. In the first turn the honest prover provides the verifier with a snapshot state of at the $(2\ell + 1 )$-th turn, which includes the state of its private register and the message register in the original protocol. In the second turn the verifier flips a fair coin and sends it to the prover. In the remaining turns they perform the forward or backward simulation of the original protocol, according to the result of the coin flip. We then go back to the $\mathsf{dQIP}$ model in Section~\ref{subsection:private} and show how to reduce the number of turns by half in the $\mathsf{dQIP}$ model by using the same argument as in the case of $\mathsf{dQIP}^{sh}$, by using two additional turns in order to share the result of the coin flip. %(since all nodes have to share the result of the coin flip in order to run the same protocol as in Section \ref{subsection:shared} that runs in the $\mathsf{dQIP}^{sh}$ model. This requires two additional turns. 
We can thus parallelize $(4\ell + 1)$-turn $\mathsf{dQIP}$ protocols to $(2\ell + 3)$-turn. Applying recursively this approach makes possible to reduce the number of turns down to $7$ (corresponding to $\ell=2$), but not lower.
After that, we focus on how to parallelize 7-turn $\mathsf{dQIP}$ protocols to 5-turn. Starting from 7-turn, we can reduce the number of turns to 5 in the $\mathsf{dQIP}^{sh}$ model, as in the $\mathsf{QIP}$ model. In $\mathsf{dQIP}p$ model we need additional two turns, so we need a different approach to turn-reduction when we start 7-turn protocols. To parallelize 7-turn to 5-turn, we use a protocol similar to the Marriott-Watrous protocol~\cite{marriott2005quantum}. Their protocol is used to show that $\mathsf{QIP}[3]\subseteq \mathsf{QMAM}$, where $\mathsf{QMAM}$ is the subclass of $\mathsf{QIP}[3]$ in which messages Arthur can send to Merlin are random bits. We construct a similar protocol, which can be used to parallelize 7-turn $\mathsf{dQIP}p$ protocols to 5-turn $\mathsf{dQIP}p$ protocols. 

%After that, we focus on how to parallelize 5-turn $\mathsf{dQIP}$ protocols to 3-turn, in Section \ref{section:3_turns}. When we are only allowed to use three turns, we cannot use interactions to share the coin flip. The nodes thus cannot share random bits at all. Therefore, instead of using the same method, we consider a protocol that simultaneously runs the forward and backward simulations. In the first turn the honest prover provides the verifier with a quantum state that can simulate the forward and backward simulations simultaneously. %two copies of the snapshot state of at the third turn of the original 5-turn protocol. 
%Each node then decides which register to use for the forward simulation uniformly at random. This randomness is necessary in order not to be fooled by an malicious prover. However, in general, the snapshot state sent by the (honest) prover may be entangled across multiple nodes. In this case, randomly selecting  registers may result in the state being different from the snapshot state. Therefore, in order to use this approach, additional conditions about the snapshot state have to be imposed.

In Section \ref{section:simulation} we then prove Theorem~\ref{theorem:main_theorem} and Theorem~\ref{theorem:main_theorem_shared}. We first discuss how to convert $\mathsf{dAM}[k](f(n))$ protocols to $\mathsf{dQIP}[k](f(n))$ protocols. This is achieved by doing all the computation of the $\mathsf{dAM}$ protocol in a reversible manner (i.e., unitary computation). Since the verification phase remains classical, the probability of being fooled is as low as the original $\mathsf{dAM}$ protocol, no matter what entangled state the malicious prover sends. This converted $\mathsf{dQIP}[k](f(n))$ protocol is then parallelized to 5-turn (3-turn in $\mathsf{dQIP}^{sh}$ model, respectively) by repeatedly using the technique we show in Section~\ref{section:general}. We also have to discuss the size of messages. Since in $\mathsf{dAM}[k](f(n))$ protocols, the private registers of the nodes are used only to store a copy of certificates, the size of the private registers of nodes in the converted $\mathsf{dQIP}[k]$ protocol is $O(f(n))$. Therefore the size of the snapshot states given by the prover is also $O(f(n))$.

In Section~\ref{section:dSWAP}, we tackle with the task to test closeness of two distributed quantum states ($\mathsf{DQCT}_N$ in Section~\ref{subsec:our_results}), and present a $\mathsf{dQIP}$ protocol for this task. The main difficulty is the implementation of the controlled SWAP gate, since there is no prior shared entanglement in the network. To resolve this issue, we utilize the protocol of Zhu and Hayashi~\cite{zhu2019efficient} to make the nodes share the GHZ state $\frac{1}{\sqrt{2}}(\ket{0^n} + \ket{1^n})$. Using the prover, we show the tests (described by some POVM measurement) in the protocol of~\cite{zhu2019efficient} can be implemented in the distributed setting.
Another difficulty is to avoid to be fooled by the malicious prover. This is achieved by carefully constructing the protocol, which ensures that the malicious prover cannot fraudulently increase the acceptance probability.

\begin{comment}
In our protocol, as in the SWAP test, the verifier accepts with probability proportional to the inner product of two inputs (if the two states are identical, the acceptance probability is 1, and if the two states are orthogonal, the acceptance probability is 
1/2). We also have to note that the message size of this protocol is $O(1)$ since the verifier needs the help of the prover only to distribute control qubits to implement the controlled-SWAP gate.
The key property of our protocol is that if the verifier accepts with high probability, then the trace distance between two input states is small. While this kind of property has already been shown even for entangled inputs~\cite{fraigniaud2021distributed}, we constructed a protocol that satisfies this property even under the existence of a potentially malicious prover.
\end{comment}

\subsection{Related Works}\label{subsec:related_work}
Although there is no previous result about distributed interactive proofs with quantum resources, Fraigniaud, Le Gall, Nishimura and Paz \cite{fraigniaud2021distributed} investigated the quantum version of \textit{randomized proof-labeling schemes} (or equivalently, $\mathsf{dMA}$ protocols). They considered the following problem: $N$-bit inputs (where $N$ is sufficiently larger than~$n$) are given to several nodes in a network and the goal is to check if all inputs are equal. They gave a $\mathsf{dQMA}$ protocol with $O(\log N)$ certificate size, and showed that any classical $\mathsf{dMA}$ protocol requires $\Omega(N)$ certificate size, which shows the superiority of quantum certification in this setting.

The study of distributed interactive proofs for some concrete problems has been developed recently, including two or three turn distributed Merlin-Arthur protocols for recognition of cographs, distance-hereditary graphs, and some
geometric intersection graph classes \cite{montealegre2021compact,jauregui2021distributed}. 

Research on quantum distributed algorithms that outperform classical distributed algorithms in several standard distributed models has been very active recently: there have been investigations showing the superiority of quantum distributed algorithm in the $\congest$ model \cite{legall2018sublinear,izumi2020quantum,censor2022quantum}, the $\clique$ model \cite{izumi2019quantum} and the $\local$ model \cite{legall2019quantum}.

\section{Definitions}\label{section:definitions}

\subsection{Distributed Interactive Proofs}\label{sub:dAM}
In this section we describe classical distributed interactive proofs, following the definition by Kol, Oshman and Saxena \cite{kol2018interactive}.

In distributed interactive proofs the verifier consists of a network represented by connected graph $G=(V,E)$ with $|V|=n$ nodes, and each node $u\in V$ is given its input label $I(u)$ where $I:V\rightarrow \{0,1\}^{*}$ is a function. We let $\mathcal{G}$ be the set of all connected graphs on vertices $V$, and $\mathcal{I}$ be the set of functions that maps $V$ to $\{0,1\}^{*}$. Define a language $\mathcal{L}$ by a subset
\begin{align*}
    \mathcal{L}\subseteq \mathcal{G}\times \mathcal{I}.
\end{align*}

Given a network configuration $(G,I)\in\mathcal{G}\times \mathcal{I}$ and a language $\mathcal{L}$, we consider an interactive protocol that consists of a series of interactions between a prover (\textit{Merlin}) and a distributed verifier (\textit{Arthur}). The goal of the protocol is to decide if $(G,I)\in \mathcal{L}$. The prover Merlin has unlimited computational power, and knows all information about $(G,I)$. The verifier Arthur is distributed, and initially each node $u$ only knows its input $I(u)$. Merlin is not trusted and tries to convince Arthur that $(G,I)\in \mathcal{L}$ by sending bit strings (certificates). Arthur provides Merlin some random queries. % in order not to be fooled.
There exist two types of randomness that can be used by Arthur: private randomness and shared randomness.
In the private randomness setting, each node can generate random bits that cannot be seen by the other nodes.
In the shared randomness setting, all random bits generated by Arthur are shared between all nodes. We denote the private randomness setting by $\mathsf{dAM}$ and the shared randomness setting by $\mathsf{dAM}^{sh}$.

A $k$-turn distributed interactive protocol (also called $k$-turn distributed Arthur-Merlin protocol in \cite{kol2018interactive}) begins with Merlin's turn if $k$ is odd, and Arthur's turn otherwise. 
%Here we consider $k$ is odd, 
If $k$ is odd, the protocol begins with Merlin's turn. In the first turn, Merlin chooses a function $c_1:\{0,1\}^*\rightarrow \{0,1\}^*$ determined by the network configuration $(G,I)$, and sends $c_1(u)$ to each node $u$.
In the second turn, Arthur picks a random string $r_2(u)$ at each node $u$, and sends them to Merlin. (In $\mathsf{dAM}^{sh}$ interactive protocols Arthur picks one random string and it can be seen by all nodes.)
This series of interactions continues for $k$ turns. More precisely, if the $j$-th turn is Merlin's turn, he sends a certificate $c_j(u)$ to each node $u$ where $c_j:\{0,1\}^*\rightarrow \{0,1\}^*$ is a function of $(G,I)$ and all of random strings $\{r_i(u)\}_{u\in V,i\in \{2,4,...,j-1\}}$ received from Arthur, and if the $j$-th turn is Arthur's, he picks a random string $r_j(u)$ at each node $u$, and sends them to Merlin. If $k$ is even then the first certificate is regarded as $c_1(u)=\emptyset$ and the protocol begins with Arthur's turn. 

The protocol completes with the verification phase. In this phase every node $u$ broadcasts a message $M_u$ to its neighbors which may depend on its input $I(u)$, random strings $u$ picked, and the certificates $u$ received. Finally, $u$ decides its output (accept or reject) by all information accumulated by $u$. Arthur accepts if and only if all nodes accept.
We say that a protocol has completeness $c$ and soundness $s$ for a language $\mathcal{L}$ if the following conditions hold for the verifier $G$ and the input label $I$:

\begin{enumerate}
    \item (Completeness) If $(G,I) \in \mathcal{L}$, then there exists a prover \textit{P} such that $\mathrm{Pr}(\text{all nodes accept})\geq c$. 
    \item (Soundness) If $(G,I) \notin \mathcal{L}$, then for any prover \textit{P}, $\mathrm{Pr}(\text{all nodes accept})\leq s$.
\end{enumerate}

As in \cite{kol2018interactive}, we define the class $\mathsf{dAM}[k](f(n))$ as the class of languages accepted by such $k$-turn distributed interactive protocols in which in each turn the prover and the verifier exchange $O(f(n))$ bits per node, and each node exchanges $O(f(n))$ bits with its neighbors during the verification procedure. The formal definition is as follows.

\begin{definition}[\cite{kol2018interactive}]\label{def:dAM}
The class $\mathsf{dAM}[k](f(n))$ is the class of languages $\mathcal{L}\subseteq \mathcal{G}\times \mathcal{I}$ that have a $k$-turn distributed Arthur-Merlin protocol 
with completeness $2/3$ and soundness $1/3$ 
satisfying the following conditions:
%which satisfies the following conditions:
\begin{itemize}
    \item At each Merlin's turn, Merlin sends certificates of $O(f(n))$ bits per node, and at each Arthur's turn, each node sends $O(f(n))$ random bits to Merlin.
    \item The size of messages exchanged between two adjacent nodes in the verification phase is $O(f(n))$ bits.
%    \item If $(G,I)\in \mathcal{L}$, the probability Arthur accepts is at least $\frac{2}{3}$.
%    \item If $(G,I)\notin \mathcal{L}$, the probability Arthur accepts is at most $\frac{1}{3}$.
\end{itemize}
\end{definition}

\subsection{Distributed Quantum Interactive Proofs}\label{sub:dQIP}
In this section we define the quantum counterpart of distributed interactive proofs, which we call distributed quantum interactive proofs. We assume the reader is familiar with the basic notions of quantum computation such that bra-ket notation of qubits, quantum circuits, and density operators (see \cite{nielsen2002quantum}, for instance, for a good reference).

Distributed quantum interactive proofs are defined similarly to the classical distributed interactive proofs of Section~\ref{sub:dAM}, but now the messages exchanged between the prover and the nodes consist of qubits, the nodes can do any (local) quantum computation, and each node can send messages consisting of qubits to its adjacent nodes. To make this rigorous and define the complexity of the protocol, we need to carefully specify how the messages are encoded using quantum registers and who owns the registers during the computation.\footnote{Since quantum information differs from classical information in several fundamental ways (in particular, quantum information cannot be copied and quantum message can share ``entanglement''), when studying quantum communication complexity or quantum distributed computation, a quantum message is represented as a quantum register (i.e., a physical system comprising multiple qubits) and the action of sending a quantum message is represented by sending this quantum register. The message size corresponds to the size of the register.}  Here is the formal definition.

\begin{definition}\label{def:dQIP}
A $k$-turn distributed quantum interactive proof ($\mathsf{dQIP}$) is a protocol between a prover and a distributed verifier who interact in the following way:
\begin{itemize}
    \item \textbf{The configuration:} The verifier consists of an $n$-node network $G=(V,E)$. Each node $u$ begins with a quantum register $\mathsf{V}_u$. We denote $\mathsf{V}$ the set of registers $\{\mathsf{V}_{u}\}_{u\in V}$. The prover begins with a quantum register $\mathsf{P}$. In addition, there is a quantum message register $\mathsf{M}_u$ for each $u$. %$\{\mathsf{M}_u\}_{u\in V}$,
    Let $\mathsf{M}$ be the set of registers $\{\mathsf{M}_{u}\}_{u\in V}$. The prover initially has the register $\mathsf{M}$ if $k$ is odd, otherwise the node $u$ initially has the register $\mathsf{M}_u$. The initial state in $\mathsf{V}$ and $\mathsf{M}$ is the all-zero pure state $\ket{0\cdots 0}$.
    \item \textbf{The interaction:} The interaction of a $\mathsf{dQIP}$ system is the repetition of prover's turn and verifier's turn. In the prover's turn, the prover performs arbitrary unitary transform denoted $P_i$ to $(\mathsf{M},\mathsf{P})$ and sends each $\mathsf{M}_u$ to~$u$ in the $i$-th turn. In the verifier's turn, the verifier can do any local (quantum) computation. More precisely, each node $u$ performs an arbitrary unitary transform denoted $V_{u,i}$ to $(\mathsf{V}_u,\mathsf{M}_u)$. Then, each node $u$ sends $\mathsf{M}_u$ to the prover in the $(i+1)$-th turn. 
    %$2i$-th turn. 
    We let $V_i=\bigotimes_{u\in V}V_{u,i}$ be the unitary transform applied by the verifier in the $i$-th turn.  
    %\item \textbf{The verification:} After the interaction phase, each node $u$ prepares registers $\mathsf{W}_{u,v}$ for $(u,v)\in E$ which are initialized to $\ket{0\cdots 0}$.
    %Then $u$ performs arbitrary unitary transform to $\mathsf{V}_u,\mathsf{M}_u$ and $\mathsf{W}_{u,v}$ for $(u,v)\in E$. After that for any $(u,v)\in E$, the two registers $\mathsf{W}_{u,v}$ and $\mathsf{W}_{v,u}$ are swapped by the SWAP gate, which is a two-qubit gate that transforms $|a\rangle|b\rangle$ to $|b\rangle|a\rangle$ for any $a,b\in\{0,1\}$. 
    %Here, we let $V_{k+1}$ be the unitary transform that is performed in the verification phase. In summary, if $k$ is odd then the interaction begins with the prover's turn, and the entire unitary transform is written by $Q=V_{k+1} P_{k} \cdots V_2 P_1$. If $k$ is even then $Q$ is written by $Q=V_{k+1} P_{k} \cdots P_2 V_1$. After that, each node $u$ performs a POVM measurement $(\Pi_{\text{acc},u},\Pi_{\text{rej},u}=I-\Pi_{\text{acc},u})$ on register $\mathsf{V}_{u},\mathsf{M}_{u}$ and $\mathsf{W}_{u,v}$ for $(u,v)\in E$ to obtain its output. Without loss of generality, we can assume $\Pi_{\text{acc},u}=\ket{0}\bra{0}\otimes I$ for all $u\in V$, i.e., node $u$ accepts the protocol iff the first qubit of register $\mathsf{V}_{u}$ is in the state $\ket{0}$.
    \item \textbf{The verification:} After the interaction phase, each node $u$ prepares registers $\mathsf{W}_{u,v}$ for $(u,v)\in E$ which are initialized to $\ket{0\cdots 0}$ and used for communication.
    Then $u$ performs an arbitrary unitary transform on the registers $\mathsf{V}_u,\mathsf{M}_u$ and all $\mathsf{W}_{u,v}$ for $(u,v)\in E$. After that, each node communicates with its neighbors, performs a measurement, and then decides reject/accept based on the outcome of the measurement (a more formal description of this step is given to the end of Section \ref{sub:dQIP}).
\end{itemize} 
\end{definition}

    Note that distributed quantum interactive proofs as defined above can simulate random bits.\footnote{Concretely, simulating one random bit can be done by using the Bell pair $\frac{1}{\sqrt{2}}|00\rangle+\frac{1}{\sqrt{2}}|11\rangle$ and keeping one qubit of the pair.} The size of the certificate sent from the prover and node $u$ at each prover's turn is the size of the register $\mathsf{M}_u$. The size of the message sent from node $u$ to the prover at each verifier's turn is also the size of the register $\mathsf{M}_u$. At the verification phase, the size of the message exchanged between node $u$ and $v$ is the size of the register $\mathsf{W}_{u,v}$. This leads to the following definition of the complexity class $\mathsf{dQIP}[k](f(n))$, as the natural quantum variant of the complexity class $\mathsf{dAM}[k](f(n))$ of Definition \ref{def:dAM}.

\begin{definition}
The class $\mathsf{dQIP}[k](f(n))$ is the class of languages $\mathcal{L}$ such that there exists a $k$-turn $\mathsf{dQIP}$ protocol for $\mathcal{L}$  with completeness $\frac{2}{3}$ and soundness $\frac{1}{3}$ satisfying the following conditions:
\begin{itemize}
    \item The size of register $\mathsf{M}_u$ for each node $u$ is $O(f(n))$. 
    \item The size of register $\mathsf{W}_{u,v}$ exchanged between $u$ and $v$ in the verification phase is $O(f(n))$ for any $(u,v)\in E$.
\end{itemize}
\end{definition}

\subsubsection{Technical Details about the Verification Phase}
We now give a more formal (and more technical) description of the last step of the verification phase in Definition \ref{def:dQIP}.
The communication and measurement operations can be specifically described as follows:
for any $(u,v)\in E$, the two registers $\mathsf{W}_{u,v}$ and $\mathsf{W}_{v,u}$ are swapped by the SWAP gate, which is a two-qubit gate that transforms $|a\rangle|b\rangle$ to $|b\rangle|a\rangle$ for any $a,b\in\{0,1\}$. 
    Here, we let $V_{k+1}$ be the unitary transform that is performed in the verification phase. If $k$ is odd then the interaction begins with the prover's turn, and the entire unitary transform is written by $Q=V_{k+1} P_{k} \cdots V_2 P_1$. If $k$ is even then $Q$ is written by $Q=V_{k+1} P_{k} \cdots P_2 V_1$. After that, each node $u$ performs a POVM measurement $(\Pi_{\text{acc},u},\Pi_{\text{rej},u}=I-\Pi_{\text{acc},u})$ on register $\mathsf{V}_{u},\mathsf{M}_{u}$ and $\mathsf{W}_{u,v}$ for $(u,v)\in E$ to obtain its output. Without loss of generality, we can assume $\Pi_{\text{acc},u}=\ket{0}\bra{0}\otimes I$ for all $u\in V$, i.e., node $u$ accepts the protocol iff the first qubit of register $\mathsf{V}_{u}$ is in the state $\ket{0}$.
\subsubsection{Variants of the Definition}
The above definition corresponds to the distributed quantum interactive proofs with private randomness where communication between nodes of the networks only happens after the interaction with the prover. This is the natural quantum analog of the definition of classical distributed interactive proofs by \cite{kol2018interactive} given in Section \ref{sub:dAM}.

A possible variant is distributed quantum interactive proofs with shared randomness, in which nodes are allowed to use shared randomness. In order to distinguish this model with the settings of private randomness, we denote it $\mathsf{dQIP}^{sh}$. We denote $\mathsf{dQIP}^{sh}[k](f(n))$ the complexity class defined for this variant similarly to Definition \ref{def:dAM}, with the additional condition that at each turn the size of (shared) random bits sent to the prover is also $O(f(n))$-bit (i.e., each node $u$ can send its message register and a random string $s$ of size $O(f(n))$, but $s$ must be the same as those of the other nodes). 

Another variant, which we call $\mathsf{dQIP}c$, is the variant where nodes can communicate with each other in the middle of interaction with the prover. While in this paper we do not focus on this variant (since the classical version did not consider communication in the middle of the interaction with the prover either), we present a general result about this natural setting in Appendix \ref{section:perfect}.  We denote $\mathsf{dQIP}c[k](f(n))$ the complexity class defined for this variant similarly to Definition \ref{def:dAM}.

%Therefore in the interaction phase node $u$ can apply arbitrary unitary transform to $\mathsf{W}_{v,u}$. Also, $\mathsf{W}_{u,v}$ and $\mathsf{W}_{v,u}$ can be swapped by the SWAP gate for $(u,v) \in E$.

\section{General Turn Reduction Technique for Distributed Quantum Interactive Proofs}\label{section:general}

In this section we show a general reduction technique to reduce the number of turns by half while keeping the soundness parameter relatively low. The complexity 
%(i.e., $\mathsf{dQIP}$ complexity) 
only increases by the size of the private register (i.e., the amount of quantum memory used for local computation at each node).

\subsection{Distributed Quantum Interactive Proofs with Shared Randomness}\label{subsection:shared}
We first consider the case of $\mathsf{dQIP}^{sh}$ model, and show the following theorem.

\begin{theorem}\label{theorem:shared}
Let $\ell \geq 1$ be an integer, $\mathcal{L}\subseteq \mathcal{G}\times\mathcal{I}$ be a language that has a $\mathsf{dQIP}^{sh}[4\ell + 1](f(n))$ protocol with completeness $c$ and soundness $s$ for some $c^2 > s$ where the protocol uses $g(n)$ space register at each node $u$ in the interaction phase. Then $\mathcal{L}$ has a 
%$(2\ell +1)$-turn 
$\mathsf{dQIP}^{sh}[2\ell + 1](f(n)+g(n))$ protocol with completeness $\frac{1+c}{2}$ and soundness $\frac{1+\sqrt{s}}{2}$.
\end{theorem}

\begin{proof}
We will prove this theorem by adapting the method halving the number of turns of quantum interactive proofs given by \cite{kempe2009using} into the distributed setting. In their method, the prover first provides the snapshot state of the message register and the private register at (almost) half of turns in the original protocol. Then the verifier flips a coin and decides to execute either a forward-simulation or a backward-simulation of the original protocol. The honest prover can perform the simulation according to the verifier's coin flip. On the other hand, due to the randomness of the verifier's choice, the malicious prover cannot fool 
%cheat to fool
the verifier. In order to implement this in the distributed setting, we only need to simulate the verifier's coin flip.

Let $\mathcal{L}$ be a language that has a $k$-turn $\mathsf{dQIP}^{sh}$ protocol where $k= 4\ell +1$ for some integer $\ell \geq 1$, and let $Q=V_{k+1}P_k\cdots V_2 P_1$ be the unitary transform that is applied to register $(\mathsf{V},\mathsf{M},\mathsf{P})$ in the protocol. We show a $(2\ell + 1)$-turn $\mathsf{dQIP}^{sh}$ protocol in Figure \ref{fig:protocol_shared}.

\begin{figure}[htbp]
\begin{mdframed}
\begin{enumerate}
	\item The prover sends $(\mathsf{V}_u,\mathsf{M}_u)$ to the node $u$ for all $u\in V$.	
	\item The network generates a random bit $r$ and sends it to the prover. %creates a Bell pair $\frac{1}{\sqrt{2}}\ket{00}+\frac{1}{\sqrt{2}}\ket{11}$ in the shared register $\mathsf{R}$.
	\item 
	\begin{enumerate}
	    \item The network applies $V_{2\ell+2}^{\dagger}$
	    to $(\mathsf{V},\mathsf{M})$ if $r=1$,
	    and sends $\mathsf{M}$.
	    \item For $j=1$ to $\ell-1$
	    %$\ell$, 
	    do the following:\\
	    Each node $u$ receives $\mathsf{M}_u$ from the prover. The network applies $V_{2\ell + 2j +2}$
	    %$V_{2\ell + 2j}$ 
	    if $r=0$, and $V_{2\ell - 2j +2}^{\dagger}$
	    %$V_{2\ell - 2j}^{\dagger}$ 
	    if $r=1$.
	    %(here we define $V_0 = I$). 
	    Then each node $u$ sends $\mathsf{M}_u$ to the prover.
	    \item Each node $u$ receives $\mathsf{M}_u$ from the prover. The network applies $V_{4\ell + 2}$
	    %$V_{4\ell + 1}$ 
	    if $r=0$, 
	    and applies $V_2^\dagger$ if $r=1$.
	    If $r = 0$, then each node $u$ performs the POVM measurement $\{\Pi_{\text{acc},u},I-\Pi_{\text{acc},u}\}$ on register $(\mathsf{V}_{u},\mathsf{M}_{u})$ and obtains its output.
	    If $r=1$, each node $u$ measures the content of $\mathsf{V}_u$ and accepts iff the qubits in $\mathsf{V}_u$ are in $\ket{0\cdots 0}$.
	\end{enumerate}
\end{enumerate}
\end{mdframed}
\caption{$(2\ell + 1)$-turn $\mathsf{dQIP}^{sh}$ protocol.}
\label{fig:protocol_shared}
\end{figure}

It is obvious that the protocol in Figure \ref{fig:protocol_shared} is a $(2\ell +1)$-turn $\mathsf{dQIP}^{sh}$ protocol. The analysis of completeness and soundness can be shown in the same way as in \cite{kempe2009using}.
\end{proof}

Applying recursively Theorem~\ref{theorem:shared} makes possible to reduce the number of turns down to 3: Let $m$ be the minimum integer that satisfies $k\leq 2^m+1$. Then, the number of turns can be reduced from $k$ to 3 by applying Theorem~\ref{theorem:shared} $m-1$ times.

\subsection{Distributed Quantum Interactive Proofs without Shared Randomness}\label{subsection:private}

Next, we consider $\mathsf{dQIP}$ protocols and show the analogous result of Theorem \ref{theorem:shared} in the $\mathsf{dQIP}$ model.
We can implement the protocol of Figure \ref{fig:protocol_shared} in the $\mathsf{dQIP}$ model by simulating the step (2) of Figure \ref{fig:protocol_shared} (the verifier's coin flip) without shared randomness. In order to simulate it, we need additional two turns:
In the first turn the prover sends the information that represents a rooted spanning tree along with the snapshot state. 
The root (denoted by $\ell$) of the spanning tree creates a Bell pair $\frac{1}{\sqrt{2}}\ket{0}_{\mathsf{M}_{\ell}}\ket{0}_{\mathsf{V}_{\ell}} + \frac{1}{\sqrt{2}}\ket{1}_{\mathsf{M}_{\ell}}\ket{1}_{\mathsf{V}_{\ell}}$ using one qubit of its message register and one qubit of its private register, then sends the message register to the prover in the second turn.
The prover in the third turn creates $\frac{1}{\sqrt{2}}\ket{0^{n}}_{\mathsf{M}} + \frac{1}{\sqrt{2}}\ket{1^{n}}_{\mathsf{M}}$ using the CNOT gate, sends one qubit of them to all nodes except the root $\ell$, and keeps one-qubit. The construction of a rooted spanning tree in 1-turn requires $\Omega(\log n)$ witness size \cite{korman2010proof}, but we can construct a rooted spanning tree in $O(1)$ witness size in 3-turn by the result of \cite{naor2020power}. Therefore the size of witnesses is unchanged without a constant factor and we obtain the following theorem.

\begin{theorem}\label{theorem_5_turn}
Let $\ell \geq 1$ be an integer, $\mathcal{L}\subseteq \mathcal{G}\times\mathcal{I}$ be a language that has a $\mathsf{dQIP}p[4\ell + 1](f(n))$ protocol with completeness $c$ and soundness $s$ for some $c^2 > s$ where the protocol uses $g(n)$ space register at each node $u$ in the interaction phase. Then $\mathcal{L}$ has a $\mathsf{dQIP}p[2\ell + 3](f(n)+g(n))$ protocol with completeness $\frac{1+c}{2}$ and soundness $\frac{1+\sqrt{s}}{2}$.
\end{theorem}

Note that applying recursively Theorem~\ref{theorem_5_turn} makes possible to reduce the number of turns down to $7$ (corresponding to $\ell=2$), but not lower: Let $\ell$ be the minimum integer that satisfies $k\leq 4\ell+1$. Starting from $k$-turn, using Theorem~\ref{theorem_5_turn}, it is reduced to $2\ell+3$. For $\ell' = \lfloor \frac{\ell}{2} \rfloor$, we have $4(\ell'+1)+1\geq 2\ell+3$. Using Theorem~\ref{theorem_5_turn} again, it is reduced down to $2(\ell'+1)+3$, which is at most $2\ell+1$ if $\ell\geq 3$. However if $\ell\leq 2$, we cannot reduce from $4\ell+1$ to $2\ell+1$ using Theorem~\ref{theorem_5_turn}, that is, the recursion stops at 7-turn. %We discuss how to overcome this barrier in Section~\ref{subsection:7_to_5}.
%\subsection{Reduction from 7-turn to 5-turn}\label{subsection:7_to_5}
In Appendix~\ref{appendix:7_to_5}, we show the following theorem, which enables us to parallelize 7-turn protocols. %The proof is found in Appendix~\ref{appendix:7_to_5}.

\begin{theorem}\label{thm:7_to_5}
Let $\mathcal{L}\subseteq \mathcal{G}\times\mathcal{I}$ be a language that has a $\mathsf{dQIP}p[7](f(n))$ protocol with completeness $c$ and soundness $s$ where the protocol uses $g(n)$ space register at each node $u$ in the interaction phase. Then $\mathcal{L}$ has a $\mathsf{dQIP}p[5](f(n)+g(n))$ protocol with completeness $\frac{1+c}{2}$ and soundness $\frac{1+\sqrt{s}}{2}$.
\end{theorem}

%%%%%%%%%%%%%%%%%%%%%%%%%%%%%%%%%%%%
\section{Quantum Simulation of Distributed Arthur-Merlin Interactive Protocols}\label{section:simulation}
%%%%%%%%%%%%%%%%%%%%%%%%%%%%%%%%%%%%
\begin{figure}[htbp]
\begin{mdframed}
\begin{enumerate}
    \item \textbf{The prover's turn:} The prover applies arbitrary unitary $U_j$ to the register $(\mathsf{P},\mathsf{M})$. Then $\mathsf{M}_u$ is sent to the node $u$.
    \item \textbf{The verifier's turn:} Each node $u$ stores the state in $\mathsf{M}_u$ to its private register $\mathsf{V}_u$ by the SWAP gate, and creates  $\frac{1}{\sqrt{2^m}}\sum_{r\in\{0,1\}^m}\ket{r}_{\mathsf{M}_u}\ket{r}_{\mathsf{V}_u}$ in ${\mathsf{M}_u}$ and a fresh part of ${\mathsf{V}_u}$. Then  $\mathsf{M}_u$ is sent to the prover.
    \item \textbf{The verification phase:} Each node measures its private register in the computational basis, then broadcasts the outcome to its neighbors, and decides the output of the protocol (accept or reject).
\end{enumerate}
\end{mdframed}
\caption{$\mathsf{dQIP}pp$ protocol simulating a $\mathsf{dAM}$ protocol.}
\label{fig:protocol_simulation}
\end{figure}

In this section we see how to convert $\mathsf{dAM}$ protocols to $\mathsf{dQIP}$ protocols, and parallelize the converted $\mathsf{dQIP}$ protocol to 5-turn. Assume the size of each witness and random bits is $m$. Let $c_j$ be a function that represents the witnesses provided by Merlin at the $j$-th turn. That is, if random bits generated by Arthur in the $i$-th turn is $r_i = r_i(u_1)r_i(u_2)\cdots r_i(u_n)\in \{0,1\}^{mn}$, $c_j(r_2,r_4,\ldots,r_{j-1})=c_j(u_1)c_j(u_2)\cdots c_j(u_n)\in \{0,1\}^{mn}$ represents the witness where $c_j(u)$ is provided to $u$. In order to simulate $k$-turn $\mathsf{dAM}$ protocols, each computation by the prover has to be converted to a form of reversible computation. Thus $c_j$ must be realized by a unitary transform
\begin{align*}
    U_{c_j}:\ket{r_2,...,r_{j-1},b}\rightarrow \ket{r_2,...,r_{j-1},b \oplus c_j}.
\end{align*}

The protocol proceeds as Figure \ref{fig:protocol_simulation}.
Let $c$ and $s$ be the completeness and the soundness of the original $\mathsf{dAM}[k]$ protocol, respectively. We show the following theorem.
\begin{theorem}\label{theorem:simulation}
The protocol in Figure~\ref{fig:protocol_simulation} has completeness $c$ and soundness $s$.
\end{theorem}
Using this theorem and the results in Section~\ref{section:general}, we can show Theorems~\ref{theorem:main_theorem} and~\ref{theorem:main_theorem_shared}. The proof can be found in Appendix~\ref{appendix:main_theorem}.

\subsection{Applications of Theorem \ref{theorem:main_theorem} and \ref{theorem:main_theorem_shared}}
We can apply Theorem~\ref{theorem:main_theorem} and Theorem~\ref{theorem:main_theorem_shared} to the following $\mathsf{dAM}$ protocols by \cite{naor2020power} (see Appendix~\ref{appendix:problems} for the definition of these problems), obtaining Corollary~\ref{corollary:applications}:

\begin{theorem}[\cite{naor2020power}]\label{th:cl1}
There exist 
\begin{itemize}
    \item a $\mathsf{dAM}[4](\log n)$ protocol for $\textsc{Asym}$,
    \item a $\mathsf{dAM}[O(1)](\log n)$ protocol for $\textsc{GNI}$,
    \item a $\mathsf{dAM}[5](\log \log n)$ protocol for $\textsc{SetEquality}$,
    \item a $\mathsf{dAM}[5](\log \log n)$ protocol for $\textsc{DSym}$.
\end{itemize}
\end{theorem}

\begin{theorem}[\cite{naor2020power}]\label{th:cl2}
There exists a constant $\delta$ such that if
a language $\mathcal{L}$ can be decided in $\mathrm{poly}(n)$ time and $n^{\delta}$ space, then $\mathcal{L}\in \mathsf{dAM}[O(1)](\log n)$.
\end{theorem}
%Applying Theorem \ref{theorem:simulation} we obtain the following corollary.
\setcounter{theorem}{2}
\begin{corollary}\label{corollary:applications}
\begin{enumerate}
\item There exist 
    \begin{itemize}
        \item a $\mathsf{dQIP}^{sh}[3](\log n)$ protocol for $\textsc{Asym}$,
        \item a $\mathsf{dQIP}^{sh}[3](\log n)$ protocol for $\textsc{GNI}$,% where $\textsc{GNI}$ is the harder version defined in Section 2,
        \item a $\mathsf{dQIP}^{sh}[3](\log \log n)$ protocol for $\textsc{SetEquality}$,
        \item a $\mathsf{dQIP}^{sh}[3](\log \log n)$ protocol for $\textsc{DSym}$.
        \item a $\mathsf{dQIP}[5](\log n)$ protocol for $\textsc{GNI}$.
    \end{itemize} 
\item There exists a constant $\delta$ such that if
a language $\mathcal{L}$ can be decided in $\mathrm{poly}(n)$ time and $n^{\delta}$ space, then $\mathcal{L}\in \mathsf{dQIP}[5](\log n)$ and $\mathcal{L}\in \mathsf{dQIP}^{sh}[3](\log n)$.
\end{enumerate}
\end{corollary}
\setcounter{theorem}{15}

%======================================
\section{Testing Closeness of Two Quantum States} \label{section:dSWAP}
%======================================
\paragraph*{Verification of the GHZ state.}
In this section we briefly explain how to show Thereom~\ref{thm:closeness}. Tehchnically, our protocol can be viewed as the distributed implmentation of the SWAP test~\cite{buhrman2001quantum}. To do this, we need to implement the controlled SWAP gate, but it is not possible by local operations at each node if the inputs are distributed since there is no prior entanglement in our setting. To resolve this issue, we create the quantum state that is called the GHZ state using the prover. Let $\ket{GHZ}$ be the $n$-qubit GHZ state
\begin{align*}
    \ket{GHZ}=\frac{1}{\sqrt{2}}(\ket{0^n}+\ket{1^n}).
\end{align*}
The detail of our approach is omitted from the main body of the paper due to space constraint. It can be found in Appendix~\ref{appendix:GHZ}.
Ultimately, we present a $\mathsf{dQIP}p$ protocol $\mathcal{P}_{GHZ}$, and show the following theorem. %See Appendix~\ref{appendix:GHZ} for the detail.
\begin{theorem}\label{thm:protocol_2}
Let $\mathcal{P}_{GHZ}$ be the protocol shown in Figure~\ref{dQIP_GHZ}.
Then $\mathcal{P}_{GHZ}$ has the following properties:
\begin{itemize}
    \item (completeness): If the prover is honest, the protocol is accepted with probability 1.
    \item (soundness): If the protocol is accepted with probability $\delta$, then the reduced state $\rho$ of the output register $\mathsf{R}_{b_{target}}(u)$ satisfies
    \begin{align*}
        \bra{GHZ}\rho \ket{GHZ} \geq 1- \varepsilon.
    \end{align*}
\end{itemize}
\end{theorem}

%=====================
\paragraph*{The $\mathsf{dQIP}p$ protocol for $\mathsf{DQCT}_N$.}
%===================== 
In Appendix~\ref{appendix:dqct}, using the protocol $\mathcal{P}_{GHZ}$, we present a protocol $\mathcal{P}_{\mathsf{DQCT}}$, which satisfies the desired conditions appeard in Theorem~\ref{thm:closeness}.

\bibliography{mybibliography}

\begin{appendix}

\section{Problems}\label{appendix:problems}
In this appendix we formally define the problems Set Equality, Graph Asymmetry, Dumbbell Symmetry and Graph Non-Isomorphism.
\begin{definition}[Set Equality~\cite{naor2020power}]
Let $G$ be a graph and $I$ be an input such that the label $I(u)$ for each node $u$ contains two lists of $\ell$ elements  $\mathcal{A}_u = \{a_{u,1},\ldots,a_{u,\ell}\}$ and $\mathcal{B}_u=\{b_{u,1},\ldots,b_{u,\ell}\}$ where $\ell \leq n$ is an integer and each element in $\mathcal{A}_u$ and $\mathcal{B}_u$ can be represented in $O(\log n)$-bit.
The language $\textsc{SetEquality}$ is the set of graphs and labels such that $\{\mathcal{A}_u\}_{u\in V} = \{\mathcal{B}_u\}_{u\in V}$ as multisets.
\end{definition}

\begin{definition}[Graph Asymmetry~\cite{naor2020power}]
The language $\textsc{Asym}$ is the set of all connected graphs that do not have a nontrivial automorphism.
\end{definition}

\begin{definition}[Dumbbell Symmetry~\cite{kol2018interactive}]
Let $m,k$ be positive integers and let $n=2m+2k+1$. An $n$-vertex connected graph $G=(\{0,1,\ldots ,n-1\},E)$ is a \textit{dumbbell graph} if it satisfies following conditions:
\begin{itemize}
    \item Let $G_0$ be the vertex-induced subgraph of $G$ on vertices $\{0,\ldots,m-1\}$ and $G_1$ be the vertex-induced subgraph of $G$ on vertices $\{m,\ldots,2m-1\}$.
    \item $G_0$ and $G_1$ are connected to each other by the following path of length $2k+2$ 
    \begin{center}
        $0-(2m)-(2m+1)-\cdots - (2m+2k)-(m)$.
    \end{center}
    \item $E$ consists of all edges in $G_0$ and $G_1$, and the path-edges.
\end{itemize}
The automorphism $\sigma$ is given as follows:
    \begin{align*}
    \sigma(i) = \left\{
\begin{array}{ll}
m+i & \text{if $i\in\{0,\ldots,m-1\}$}\\
i-m & \text{if $i\in\{m,\ldots,2m-1\}$}\\
4m+2k-i & \text{if $i\in\{2m,\ldots,2m+2k\}$}
%2k - i & \text{if $i\in\{2m+k+1,\ldots,2m+2k\}$}
\end{array}
\right.
    \end{align*}
The language $\textsc{DSym}$ is the set of all dumbbell graphs $G$ such that $\sigma(G)$ is isomorphic to $G$. 
\end{definition}

\begin{definition}[Graph Non-Isomorphism~\cite{kol2018interactive}]
The language $\textsc{GNI}$ is the set of all pairs of graphs $(G_0,G_1)$ where $G_0$ is not isomorphic to $G_1$. %There exist two different definitions of $\textsc{GNI}$. The first one, which was introduced by \cite{kol2018interactive}, is harder and in this problem 
We assume that the communication graph is $G_0$, and nodes cannot communicate on $G_1$-edges. %We call this version $\textsc{GNI}^{hard}$. The second one \cite{naor2020power} is easier, and we assume that the communication graph is the union of $G_0$ and $G_1$. 
\end{definition}
 %Hence the input $I(u)$ to each node $u$ contains its incident $G_0$-edges and $G_1$-edges.
 
 In \cite{naor2020power}, it was shown that $\textsc{SetEquality} \in \mathsf{dAM}[2](\log n)$, $\textsc{SetEquality} \in \mathsf{dAM}[4](\log \log n)$, $\textsc{Asym} \in \mathsf{dAM}[4](\log n)$, $\textsc{DSym} \in \mathsf{dAM}[4](\log \log n)$, $\textsc{GNI} \in \mathsf{dAM}[k](\log n)$ for some constant $k>4$. \footnote{If we add the condition that the nodes can communicate on $G_1$-edges, there is a $\mathsf{dAM}[4](\log n)$ protocol~\cite{naor2020power}.} For $\textsc{GNI}$, there also exists a $\mathsf{dAM}[4](n\log n)$ protocol showed by \cite{kol2018interactive}.

\section{Quantum information}\label{sec:quantum_information}
We assume that the readers are familiar with basic concepts of quantum information such as density matrices, measurements, and quantum circuits (See, e.g., \cite{NC00,Wat18book,Wil17}).

Let $\mathcal{H}$ be a finite dimensional Hilbert space, and $\rho,\sigma$ be any quantum states in $\mathcal{H}$. The fidelity of two quantum states $\rho,\sigma$ is defined as $F(\rho,\sigma) = \mathrm{tr}\left[\sqrt{\sqrt{\rho}\sigma\sqrt{\rho}}\right]$. Note that for two pure states $\rho = \ket{\psi_{\rho}}\bra{\psi_{\rho}}$ and $\sigma = \ket{\psi_{\sigma}}\bra{\psi_{\sigma}}$ we have $F(\rho,\sigma) = |\braket{\psi_{\rho}|\psi_{\sigma}}|$. Let $\mathrm{dist}(\rho,\sigma) = \frac{1}{2}\|\rho-\sigma\|_{\mathrm{tr}}$ be the trace distance of $\rho,\sigma$, where $\| A \|_{\mathrm{tr}} = \mathrm{tr}\sqrt{A^{\dagger}A}$. For two pure states $\ket{\psi}\bra{\psi},\ket{\phi}\bra{\phi}$ we denote $\mathrm{dist}(\ket{\psi},\ket{\phi})$.
Here we summarize some useful inequalities about the fidelity and the trace distance, which are used multiple times in this paper.

\begin{lemma}\label{lem:inequalities}
For any quantum states $\rho,\sigma,\xi$ in $\mathcal{H}$, we have
\begin{enumerate}
    \item {\rm\cite{fuchs1999cryptographic}:} $
            1-F(\rho,\sigma) \leq \mathrm{dist}(\rho,\sigma) \leq \sqrt{1-F(\rho,\sigma)^2},
        $
    \item {\rm\cite{nayak2003bit,spekkens2001degrees}:} $
        F(\rho,\sigma)^2  + F(\xi,\sigma)^2 \leq 1 +     F(\rho,\xi).
        $
    %\item {\rm\cite{wildelecture}:} $F(\ket{\psi}\bra{\psi},\rho)^2  + F(\ket{\psi}\bra{\psi},\sigma)^2 \leq \mathrm{dist}(\rho,\sigma).$
\end{enumerate}

\end{lemma}

\section{Proof of Theorem~\ref{thm:7_to_5}}\label{appendix:7_to_5}

Fix the input $x$ and a $\mathsf{dQIP}p[7](f(n))$ protocol $\pi$ for $\mathcal{L}$ described by a sequence of unitaries $P_1,V_2,P_3,V_4,P_5,V_6,P_7,V_8$ in this order, which has completeness $c$ and soundness $s$. Our converted protocol is shown in Figure~\ref{fig:protocol_7_to_5} (we call this protocol $\pi'$). We denote $\mathsf{R}_1=\{\mathsf{R}_{u,1}\}_{u\in V}$ and $\mathsf{R}_2=\{\mathsf{R}_{u,2}\}_{u\in V}$. In $\pi'$, we consider the entire register is $(\mathsf{P,R_1,R_2})$ where $\mathsf{P}$ is the prover's private register: Initially, there is no verifier's private register, and after receiving $\mathsf{R}_1$, the private register of each node $u$ is $\mathsf{R}_{u,1}$. 
Here we analyze the completeness and the soundness of $\pi'$.
Define two quantum states $\ket{\psi_4} = V_4P_3V_2P_1\ket{0\cdots 0}_{(\mathsf{P,R_1,R_2})}$ and $\ket{\psi_5} = P_5\ket{\psi_4}$ and their reduced states
$\sigma_1=\mathrm{tr}_{(\mathsf{P},\mathsf{R}_2)}(\ket{\psi_4}\bra{\psi_4})=\mathrm{tr}_{(\mathsf{P},\mathsf{R_2})}(\ket{\psi_5}\bra{\psi_5})$, $\sigma_2=\mathrm{tr}_{(\mathsf{P},\mathsf{R}_1)}(\ket{\psi_4}\bra{\psi_4})$, and $\sigma_3=\mathrm{tr}_{(\mathsf{P},\mathsf{R}_1)}(\ket{\psi_5}\bra{\psi_5})$. (Here we abuse the notation by thinking unitaries $V_i$ act on both $(\mathsf{M,V})$ and $(\mathsf{R_1,R_2})$, and also unitaries $P_i$ act on both $(\mathsf{P,M})$ and $(\mathsf{P,R_2})$ since they have the same size.)
\par
\noindent\textbf{Proof of completeness:}
Assume that $x\in \mathcal{L}$. The honest prover does the following.
\begin{itemize}
    \item \textbf{Turn 1:} Send $\sigma_1$.
    \item \textbf{Turn 3:} Broadcast $b$. If $b=0$, send $\sigma_3$. If $b=1$, send $\sigma_2$.
    \item \textbf{Turn 5:} If $b=0$, apply $P_7$. If $b=1$, apply $P_3^{\dagger}$.
\end{itemize}
After Step 5 of Figure~\ref{fig:protocol_7_to_5}, if $b=0$, the entire quantum state is $V_8P_7V_6P_5V_4P_3V_2P_1\ket{0\cdots 0}$ and if $b=1$, the entire quantum state is $P_1\ket{0\cdots 0}=(P_1\ket{0\cdots 0}_{(\mathsf{P,R_2})})\otimes \ket{0\cdots 0}_{\mathsf{R}_1}$. Thus the acceptance probability of $\pi'$ is $\frac{1+c}{2}$.

\par
\noindent\textbf{Proof of soundness:}
Assume that $x\notin \mathcal{L}$. 
Let $\ket{\psi}$ be the initial state in $(\mathsf{P,R_1,R_2})$, that is, in Turn 1 of $\pi'$, the verifier receives the register $\mathsf{R_1}$ and its reduced state is $\mathrm{tr}_{(\mathsf{P,R_2})}(\ket{\psi}\bra{\psi})$.
Assume that, when the random bit in Turn 2 is $b=i$, the prover applies $U_i\otimes I_{\mathsf{R_1}}$ and sends $\mathsf{R_2}$ in Turn 3, and applies $W_i\otimes I_{\mathsf{R_1}}$ and sends $\mathsf{R_2}$ in Turn 5. 
Define unitaries $Q_0$ and $Q_1$ by $Q_0=(I_{(\mathsf{P,R_2})} \otimes V_8)(W_0\otimes I_{\mathsf{R_1}})(I_{(\mathsf{P,R_2})} \otimes V_6)(U_0\otimes I_{\mathsf{R_1}})$ and $Q_1=(I_{(\mathsf{P,R_2})} \otimes V_2^{\dagger})(W_1\otimes I_{\mathsf{R_1}})(I_{(\mathsf{P,R_2})} \otimes V_4^{\dagger})(U_1\otimes I_{\mathsf{R_1}})$, and let 
\begin{align*}
    \ket{\alpha}= \frac{1}{\| \Pi_{acc}Q_0 \ket{\psi} \|}\Pi_{acc}Q_0 \ket{\psi} \text{ and }
    \ket{\beta}= \frac{1}{\| \Pi_{init}Q_1 \ket{\psi} \|}\Pi_{init}Q_1 \ket{\psi},
\end{align*}
where $\Pi_{acc}$ is the projection onto the acceptance state of $\pi$, and $\Pi_{init} = I_{(\mathsf{P,R}_2)}\otimes \ket{0\cdots 0}\bra{0 \cdots 0}_{\mathsf{R_1}}$.

Let $p_i$ be the acceptance probability of $\pi'$ when the random bit in Turn 2 is $b=i$. Then we have
\begin{align*}
    p_0 &= \| \Pi_{acc}Q_0\ket{\psi}\|^2 =  \frac{1}{\| \Pi_{acc}Q_0 \ket{\psi} \|} |\bra{\psi}Q_0^{\dagger}\Pi_{acc}Q_0\ket{\psi}|^2
    = F(\ket{\alpha}\bra{\alpha},Q_0\ket{\psi}\bra{\psi}Q_0^{\dagger})^2\\
    &= F(Q_0^{\dagger}\ket{\alpha}\bra{\alpha}Q_0,\ket{\psi}\bra{\psi})^2,\\
    p_1 &= \| \Pi_{init}Q_1\ket{\psi}\|^2 =  \frac{1}{\| \Pi_{init}Q_1 \ket{\psi} \|} |\bra{\psi}Q_1^{\dagger}\Pi_{init}Q_1\ket{\psi}|^2
    = F(\ket{\beta}\bra{\beta},Q_1\ket{\psi}\bra{\psi}Q_1^{\dagger})^2\\
    &= F(Q_1^{\dagger}\ket{\beta}\bra{\beta}Q_1,\ket{\psi}\bra{\psi})^2.
\end{align*}
Therefore, using Lemma~\ref{lem:inequalities} the acceptance probability $p_{acc}$ of $\pi'$ is bounded by
\begin{align*}
    p_{acc}=\frac{1}{2}(p_0+p_1) &\leq \frac{1}{2}(1+F(Q_0^{\dagger}\ket{\alpha}\bra{\alpha}Q_0,Q_1^{\dagger}\ket{\beta}\bra{\beta}Q_1)) \\
    &=\frac{1}{2}(1+F(\ket{\alpha}\bra{\alpha},Q_0Q_1^{\dagger}\ket{\beta}\bra{\beta}Q_1Q_0^{\dagger})).
\end{align*}
We also have
\begin{align*}
    F(\ket{\alpha}\bra{\alpha},Q_0Q_1^{\dagger}\ket{\beta}\bra{\beta}Q_1Q_0^{\dagger}) = |\bra{\alpha}Q_0Q_1^{\dagger}\ket{\beta}|
    = |\bra{\alpha}\Pi_{acc}Q_0Q_1^{\dagger}\ket{\beta}|
    \leq \| \Pi_{acc}Q_0Q_1^{\dagger}\ket{\beta} \|
\end{align*}
from the fact that $\Pi_{acc}\ket{\alpha} = \ket{\alpha}$. 
The reduced state of $\ket{\beta}$ satisfies $\mathrm{tr}_{(\mathsf{P,R_2})}(\ket{\beta}\bra{\beta}) = \ket{0\cdots 0}\bra{0\cdots 0}_{\mathsf{R_1}}$ since $\Pi_{init}\ket{\beta} = \ket{\beta}$.
Therefore, from the soundness of $\pi$, for any $U_0,U_1,W_0,W_1$ acting on $(\mathsf{P,R_2})$,
\begin{align*}
    &\| \Pi_{acc}(I_{(\mathsf{P,R_2})} \otimes V_8)(W_0\otimes I_{\mathsf{R_1}})(I_{(\mathsf{P,R_2})} \otimes V_6)(U_0U_1^{\dagger}\otimes I_{\mathsf{R_1}})  (I_{(\mathsf{P,R_2})} \otimes V_4)(W_1\otimes I_{\mathsf{R_1}})(I_{(\mathsf{P,R_2})} \otimes V_2)\ket{\beta} \|^2 \\
    &=\| \Pi_{acc}Q_0Q_1^{\dagger}\ket{\beta} \|^2   \leq s.
\end{align*}
Thus we have $p_{acc}\leq \frac{1}{2} + \frac{\sqrt{s}}{2}$, which completes the proof of soundness.

\begin{figure}[htbp]
\begin{mdframed}
\begin{enumerate}
    \item \textbf{Turn 1:} The prover gives a $g(n)$-qubit quantum register $\mathsf{R}_{u,1}$ to each node $u$. The prover chooses arbitrary one node as a leader node $\mathsf{leader}$.
    \item \textbf{Turn 2:} $\mathsf{leader}$ chooses a bit $b\in\{0,1\}$ uniformly at random and sends it to the prover.
    \item \textbf{Turn 3:} The prover sends one bit $b_u$ and a $f(n)$-qubit quantum register $\mathsf{R}_{u,2}$ to each node $u$.
    \item \textbf{Turn 4:}
    If $b_u=0$, each node $u$ applies $V_{u,6}$ to $(\mathsf{R}_{u,1},\mathsf{R}_{u,2})$. If $b_u=1$, $u$ applies $V_{u,4}^{\dagger}$. $u$ sends $\mathsf{R}_{u,2}$ to the prover.
    \item \textbf{Turn 5:} The prover sends $\mathsf{R}_{u,2}$ to each node $u$. If $b_u=0$, each node $u$ applies $V_{u,8}$ to $(\mathsf{R}_{u,1},\mathsf{R}_{u,2})$. If $b_u=1$, $u$ applies $V_{u,2}^{\dagger}$. 
    
    \item \textbf{The verification phase:} If $b_u=0$, $u$ does the same verification as in the original 7-turn protocol. $u$ outputs "accept" iff the output of the original protocol is "accept" and all random bits $b_v$ where $v\in N(u)$ are the same as $b_u$.
    If $b_u=1$, $u$ outputs "accept" iff the register $\mathsf{R}_{u,1}$ is set to all-zero state and all random bits $b_v$ where $v\in N(u)$ are the same as $b_u$.
\end{enumerate}
\end{mdframed}
\caption{$\mathsf{dQIP}p[5](f(n)+g(n))$ protocol $\pi'$.}
\label{fig:protocol_7_to_5}
\end{figure}

%%%%%%%%%%%%%%%%%%%%%%%%%%%%%
\section{Proofs of Theorem~\ref{theorem:main_theorem} and Theorem~\ref{theorem:main_theorem_shared}.}\label{appendix:main_theorem}
%%%%%%%%%%%%%%%%%%%%%%%%%%%%%

We first prove Theorem~\ref{theorem:simulation}.
\begin{proof}
\noindent\textbf{Completeness}: Assume that $(G,I) \in \mathcal{L}$ and the prover receives the $\mathsf{M}_u$ part of the quantum state $\frac{1}{\sqrt{2^m}}
\sum_{r_{j-1}(u)\in\{0,1\}^m}
\ket{r_{j-1}(u)}_{\mathsf{M}_u}
\ket{r_{j-1}(u)}_{\mathsf{V}_u}$ 
from the node $u$ in the $(j-1)$-th turn.
At the $j$-th turn the honest prover applies the SWAP gate to $(\mathsf{P},\mathsf{M})$, obtaining
\begin{align*}
    \frac{1}{\sqrt{2^{\frac{j-1}{2}mn}}}
    \Biggl(\sum_{r_2,r_4,...,r_{j-1}\in \{0,1\}^{mn}}
    \ket{r_2,r_4,...,r_{j-1}}_{\mathsf{P}}
    \ket{0}_{\mathsf{M}}
    \ket{c_1,r_2,c_3,r_4,...,c_{j-2},r_{j-1}}_{\mathsf{V}}
    \Biggr).
\end{align*}
Then, the prover also applies $U_{c_j}$ to $(\mathsf{P},\mathsf{M}_u)$, obtaining the following state
\begin{align*}
    &\frac{1}{\sqrt{2^{\frac{j-1}{2}mn}}}
    \Biggl(\sum_{r_2,r_4,...,r_{j-1}\in \{0,1\}^{mn}}
    \ket{r_2,r_4,...,r_{j-1}}_{\mathsf{P}}
    \ket{c_j}_{\mathsf{M}}
    \ket{c_1,r_2,...,c_{j-2},r_{j-1}}_{\mathsf{V}}
    \Biggr)\\
    &=  \frac{1}{\sqrt{2^{\frac{j-1}{2}mn}}}\Biggl(\sum_{r_2,r_4,...,r_{j-1}\in \{0,1\}^{mn}}
    \ket{r_2,r_4,...,r_{j-1}}_{\mathsf{P}}
    \bigotimes_{u\in V}
    \ket{c_j(u)}_{\mathsf{M}_u}
    \ket{c_1,...,r_{j-1}}_{\mathsf{V}}
    \Biggr).
\end{align*}
At the verification phase, the quantum state in $\mathsf{V}$ is a mixed state
\begin{align*}
    %\frac{1}{\sqrt{2^{\frac{k-1}{2}mn}}}
    \frac{1}{2^{\frac{k-1}{2}mn}}
    \sum_{r_2,...,r_{k-1}\in \{0,1\}^{mn}}
    \ket{c_1,r_2,c_3,...,r_{k-1},c_k}
    \bra{c_1,r_2,c_3,...,r_{k-1},c_k}_{\mathsf{V}},
\end{align*}
and $u$ obtains one of the state $\ket{c_1(u),r_2(u),c_3(u),...,r_{k-1}(u),c_k(u)}$ uniformly at random as the outcome of its measurement. Then $u$ broadcasts the outcome to its adjacent nodes. From the completeness of the original $\mathsf{dAM}[k]$ protocol the acceptance probability of this verification phase is at least $c$. \par

\noindent\textbf{Soundness}: Assume that $(G,I) \notin \mathcal{L}$. 
%We label the nodes as $\{u_i|i\in [n]\}$.
%Let $P_j$ be the unitary transform that applied by the prover in the $j$-th turn. 
Since a malicious prover may use some other unitary instead of $U_{c_j}$ at the $j$-th turn. 
Let $\sum_{r_{j-1}\in\{0,1\}^{mn}} \ket{r_{j-1}}_{\mathsf{M}}\ket{r_{j-1}}_{\mathsf{V}}$ be the Bell pairs created by the verifier in the $(j-1)$-th turn. 
The witness provided by the prover in the $j$-th turn is stored into the private register $\mathsf{V}$. 
In the verification phase, 
node $u_i$ obtains 
%one of 
$\ket{x_i,r_{j-1}(u)}$ 
for some $x_i\in\{0,1\}^m$ 
as the outcome of its measurement.  
From the soundness of $\mathsf{dAM}$ protocols, the original protocol is accepted for at most $s$ of all random strings generated by Arthur.  On the other hand, each node $u$ obtains $\ket{r_2(u),r_4(u),...,r_{k-1}(u)}$ uniformly at random by its measurement regardless of the prover's action. Thus the acceptance probability of this $\mathsf{dQIP}$ protocol is at most $s$.
\end{proof}

Now we are ready to show how to parallelize the protocol to 5-turn.\par
\noindent{\textbf{Proofs of Theorem \ref{theorem:main_theorem} and Theorem \ref{theorem:main_theorem_shared}}}
Assume without loss of generality $k=4\ell + 1$ for $\ell \geq 1$ and $m=f(n)$.
For any $\mathsf{dAM}[k](m)$ protocol, we have a $\mathsf{dQIP}[k](m)$ protocol which simulates it using Theorem \ref{theorem:simulation}. Applying Theorem \ref{theorem_5_turn} we can parallelize it to a $\mathsf{dQIP}pp[2\ell + 3](m)$ protocol with completeness $\frac{1+c}{2}$ and soundness $\frac{1+\sqrt{s}}{2}$. 
Note that we can assume $g(n)=O(m)$ since
the register provided by the honest prover at the first turn of the parallelized protocol contains a $(4\ell + 1)mn$-qubit state 
in registers $\mathsf{V}$ and $\mathsf{M}$ 
such that the total state 
is represented as
\begin{align*}
    \frac{1}{\sqrt{2^{(\ell + 1)mn}}}\sum_{r_2,\ldots,r_{2\ell + 2}\in\{0,1\}^{mn}}
    \ket{r_2,\ldots,r_{2\ell+2}}_{\mathsf{P}}
    \ket{c_1,r_2,\ldots,c_{2\ell + 1},r_{2\ell + 2},0^{mn},\cdots,0^{mn}}_{(\mathsf{V},\mathsf{M})}%\otimes \sum_{r_2(u),r_4(u)\in\{0,1\}^m}\ket{c_1(u),r_2(u),c_3(u),r_4(u)}_{(\mathsf{V}_{u,2},\mathsf{M}_{u,2})}.
\end{align*}
and each node $u$ receives its reduced state of $(4\ell + 1)m=O(m)$-qubit on $(\mathsf{V}_u,\mathsf{M}_u)$.
Thus the size of witnesses and the size of messages in the verification phase are both $O(m)$.
We assume that the parameters $c$ and $s$ of the original $\mathsf{dAM}[k](m)$ protocol are $c = 1-\varepsilon$ and $s=\delta$ for small constant $\varepsilon>0$ and $\delta>0$ since we can use the standard technique of parallel repetition by \cite{crescenzi2019trade} (the protocol is executed in parallel a constant number of times, and the leader node, which is determined by the prover as a root of a spanning tree, adopts the majority of the outcomes in the verification phase). Note that the witness size does not change by parallel repetition since the protocol has $k\geq 3$ turn and the construction of a spanning tree can be done with $O(1)$-size witnesses using three turns~\cite{naor2020power}. 
Now we assume that the completeness and the soundness of the original $\mathsf{dAM}[k](m)$ protocol are $c = 1-\frac{1}{12a^2}$ and $s=\frac{1}{12a^2}$ for $a=k-1$, respectively. By Theorem~\ref{theorem:simulation}, the converted $\mathsf{dQIP}[k](m)$ protocol has the same completeness and soundness. By Theorem~\ref{theorem_5_turn}, which reduces the number of turns down to 7, and Theorem~\ref{thm:7_to_5}, which reduces the number of turns down to 5, and the same analysis as Lemma 4.2 of~\cite{kempe2009using}, we get a $\mathsf{dQIP}[5](m)$ protocol with completeness $1-\frac{2c}{k-1} = 1-\frac{1}{6a^3}$ and soundness $1-\frac{1-s}{(k-1)^2}<1-\frac{1}{2a^2}$. Then we use another parallel repetition for quantum interactive protocols developed in \cite{gutoski2010quantum}. (the parallel repetition in~\cite{gutoski2010quantum} accepts iff all outcomes in repetitions are "accept". See Theorem 4.9 in~\cite{gutoski2010quantum}.) More precisely, using $2a^3$ time repetitions the completeness and soundness become $\bigl(1-\frac{1}{6a^3}\bigr)^{2a^3} > 1-\frac{1}{3}=\frac{2}{3}$ and $\bigl(1-\frac{1}{2a^2}\bigr)^{2a^3} < \frac{1}{e^a}< \frac{1}{3}$, which completes the proof of Theorem~\ref{theorem:main_theorem}. \par
 In the case of $\mathsf{dQIP}^{sh}$, we can parallelize a $k=(4\ell + 1)$-turn protocol to a $(2\ell + 1)$-turn protocol by using Theorem~\ref{theorem:shared}. Therefore Theorem~\ref{theorem:main_theorem_shared} can be shown similarly to the proof of Theorem~\ref{theorem:main_theorem} by applying Theorem~\ref{theorem:shared}, instead of Theorem~\ref{theorem_5_turn} and Theorem~\ref{thm:7_to_5}.\qed

%%%%%%%%%%%%%%%%%%%%%%%%%%%%%%%%%%
\section{Proof of Theorem~\ref{thm:protocol_2}}\label{appendix:GHZ}
%%%%%%%%%%%%%%%%%%%%%%%%%%%%%%%%%%

We use the following characterization of the GHZ state, which states the GHZ state is \textit{locally equivalent} to a certain graph state.
\begin{lemma}\label{lem1}
\begin{align*}
    (I\otimes H^{\otimes n-1})\ket{GHZ} = \ket{S_n}
\end{align*}
where $\ket{S_n}$ is the graph state of a $n$-node star graph $S_n$ with the first qubit corresponds to the central node of $S_n$.
\end{lemma}
\begin{proof}
Let $Z= \ket{0}\bra{0}-\ket{1}\bra{1}$ be the Pauli Z gate and $CZ_{(i,j)}$ be the controlled-Z gate acting on $i$-th qubit and $j$-th qubit.
\begin{align*}
    \ket{S_n} &= \Pi_{j=2}^n CZ_{(1,j)}\ket{+}_1\ket{+}_2\cdots\ket{+}_n\\
    &= \frac{1}{\sqrt{2}} (\Pi_{j=2}^n CZ_{(1,j)}\ket{0}_1\ket{+}_2\cdots\ket{+}_n + \Pi_{j=2}^n CZ_{(1,j)}\ket{1}_1\ket{+}_2\cdots\ket{+}_n)\\
    &= \frac{1}{\sqrt{2}} (\ket{0}_1\ket{+}_2\cdots\ket{+}_n + \ket{1}_1\ket{-}_2\cdots\ket{-}_n)\\
    &= (I\otimes H^{\otimes n-1})\ket{GHZ}.
\end{align*}
\end{proof}

%\subsection{The protocol to verify graph states of 2-colorable graphs}

Assume that for $i\in [n]$, the $i$-th qubit of $\ket{GHZ}$ is owned by the $i$-th node of an $n$-node network.
We construct a $\mathsf{dQIP}$ protocol to verify the GHZ state in this setting. By Lemma~\ref{lem1}, $\ket{GHZ}$ is locally equivalent to the graph state of a 2-colorable graph $S_n$. 
In \cite{zhu2019efficient}, Zhu and Hayashi showed a protocol to verify graph states of 2-colorable graphs, which is called the coloring protocol. Thus we can leverage this to verify the star graph $S_n$. Once we verified the graph state of $S_n$, it can be transformed to the GHZ state without communication. 

Here we briefly explain their protocol. Let $A_0,A_1\subseteq V$ be the coloring of an $2$-colorable graph. Note that $A_0$ and $A_1$ are distinct with each other, and $A_0\cup A_1=V$. In the coloring protocol, there are two measurements $\{P_0,I-P_0\}$ and $\{P_1,I-P_1\}$, and the verifier performs one of the measurements uniformly at random. Here, $P_i$ is given by
\begin{align*}
    P_i = \prod_{u\in A_i}\Bigl(
      \frac{1+ K_u}{2}
    \Bigr)
\end{align*}
where
\begin{align*}
    K_u = X_u\otimes \prod_{v\in N(u)}Z_{v}.
\end{align*}
%Here $CZ_{u,v}$ is the controlled-Z operation on the qubits corresponds to $u$ and $v$. 
Therefore the verification operator $\Omega$ is written as
$
    \Omega = \frac{1}{2}(P_0+P_1).
$
%The test described by $P_i$ is performed as follows: 
%\begin{itemize}
%    \item All qubits associated with a
%given set $A_i$ are measured in the $X$ basis (i.e., the measurement $\{\ket{+}\bra{+},\ket{-}\bra{-}\}$),
%\item the other qubits ($V\backslash A_i$) are measured in the $Z$ basis (i.e., the measurement $\{\ket{0}\bra{0},\ket{1}\bra{1}\}$).
%\end{itemize}
The eigenstate of $\Omega$ with eigenvalue 1 is  stabilized by $K_u$ for all $u\in V$. Since $\ket{G}$ is the unique state that is stabilized by all $K_u$, the only state that can pass all tests $P_i$ with probability 1 is the target $2$-colorable graph state $\ket{G}$. 

Suppose that the verifer is given $n$-qubit registers $\mathsf{R}_1,...,\mathsf{R}_{N+1}$. The protocol of Zhu and Hayashi, which is refered as $\mathcal{P}_{ZH}$ in this paper, is shown in Figure~\ref{ZH}. The main result of~\cite{zhu2019efficient} is the following theorem.

\begin{theorem}[Zhu and Hayashi~\cite{zhu2019efficient}]\label{thm:ZH}
Let $G=(V,E)$ be a 2-colorable graph and $A_0,A_1$ be the coloring of $V$. Let $\mathcal{P}_{ZH}$ be the protocol shown in the Figure \ref{ZH} where $P_i = \prod_{u\in A_i}\Bigl(
      \frac{1+ K_u}{2}
    \Bigr)$ for $i\in \{0,1\}$.
For any constant parameters $\varepsilon,\delta$, there exists a constant $N=\Theta(\frac{1}{\varepsilon}\log\frac{1}{\delta})$ such that if $\mathcal{P}_{ZH}$ is accepted with probability $\delta$, then the reduced state $\rho$ of regsiter $\mathsf{R}_{N+1}$ satisfies
    \begin{align*}
        \bra{G}\rho \ket{G} \geq 1- \varepsilon.
    \end{align*}
\end{theorem}

\begin{figure}[htbp]
\begin{mdframed}
\begin{enumerate}
    \item Let $0<\varepsilon<1$, $0<\delta<1$ be two positive numbers and $N$ be a positive integer. The prover sends $n$-qubit registers $\mathsf{R}_1,...,\mathsf{R}_{N+1}$. 
    \item The verifier applies a random permutation on $\mathsf{R}_1,...,\mathsf{R}_{N+1}$. 
    For each $i\in [N]$, the verfier performs a test described as  $\Omega$ %$p\Omega + (1-p)I$ 
    on the register $\mathsf{R}_i$ where $\Omega = \frac{1}{2}(P_0+P_1)$. %and $p=1/e$.
    If all tests are passed, the verifier accepts and outputs the register $\mathsf{R}_{N+1}$.
\end{enumerate}
\end{mdframed}
\caption{The protocol of Zhu and Hayashi}
\label{ZH}
\end{figure}

\begin{figure}[htbp]
\begin{mdframed}
\begin{enumerate}
	\item \textbf{Turn 1-3:} Construct a spanning tree rooted at $\mathsf{leader}$.
    \item \textbf{Turn 1:} Let $0<\varepsilon<1$, $0<\delta<1$ be two positive numbers and $N$ be a positive integer that appeared in Theorem~\ref{thm:ZH}. The prover sends $n$-qubit registers $\mathsf{R}_1,...,\mathsf{R}_{N+1}$ (i.e., sends $\mathsf{R}_1(u),...,\mathsf{R}_{N+1}(u)$ to each node $u$).
    \item \textbf{Turn 2:} $\mathsf{leader}$ prepares a random string $b_{test}\in \{0,1\}^N$ and $b_{target}\in \{0,1\}^{\lceil \log N \rceil + 1}$ which corresponds to one element in $[N]$, then sends $b_{test}$ and $b_{target}$ to the prover. W.l.o.g., we assume that $b_{target}=N+1$ in the remaining turns.
    \item \textbf{Turn 3:} The prover sends $b_{test}(u)\in\{0,1\}^N$ and $b_{target}(u)\in \{0,1\}^{\lceil \log N \rceil + 1}$ to each node $u$.
    
    \item \textbf{Turn 4:} Let $b_{test}^i(u)$ be the $i$-th bit of $b_{test}(u)$. For $i\in [N]$, each node $u$ except $\mathsf{leader}$ measures all registers owned by $u$ except $\mathsf{R}_{b_{target}(u)}(u)$ in the Z basis if $b_{test}^i(u)=0$, otherwise in the X basis. $\mathsf{leader}$ measures all registers owned by $\mathsf{leader}$ except $\mathsf{R}_{b_{target}}(\mathsf{leader})$ in the Z basis if $b_{test}^i=0$, otherwise in the X basis. 
    Each node $u$ sends the outcome $o_u$ as an $N$-bit string (if $\ket{0}$ or $\ket{+}$ is measured, $o^i_u$ is 0, otherwise $o^i_u$ is 1) to the prover. If $u$ is not $\mathsf{leader}$, $u$ applies the Hadamard gate to $\mathsf{R}_{b_{target}}(u)$. 
    \item \textbf{Turn 5: }The prover sends N-bit $s_u\in \{0,1\}^N$ to each node $u$. 
    \item (Verification Phase)
    Each node $u$ sends $b_{test}(u)$, $b_{target}(u)$, $s_u$ and its outcome $o_u$ of the measurement to the neighbors, and checks the following for all $i\in [N]$.
    \begin{enumerate}
        \item Check if $b_{test}^i(u)$, $b_{target}(u)$ are the same with the neighbors' messages.
        \item Verify the spanning tree $T$ constructed at step 1.
        \item If $b_{test}^i=0$, check $\underset{\text{$v:$ the child of $u$ in $T$}}{\sum}s^i_v + o^i_u$ is equal to $s^i_u$ modulo 2. \\
        $\mathsf{leader}$ checks $\underset{\text{$v:$ the child of $\mathsf{leader}$ in $T$}}{\sum}s^i_v + o_{\mathsf{leader}}^i\text{ modulo 2} = 0$. 
        \item If $b_{test}^i=1$, check the value $o^i_u$ is the same as those of neighbors.
    \end{enumerate}
    If above conditions hold, $u$ accepts the protocol and outputs $\mathsf{R}_{b_{target}}$. 
    %The node $\mathsf{leader}$ accepts if the above conditions hold and $s_{\mathsf{leader}} = 0$.
\end{enumerate}
\end{mdframed}
\caption{$\mathsf{dQIP}$ protocol $\mathcal{P}_{GHZ}$ for verification of the GHZ state}
\label{dQIP_GHZ}
\end{figure}

We now construct a distributed implementation of $\mathcal{P}_{ZH}$ as a $\mathsf{dQIP}$ protocol $\mathcal{P}_{GHZ}$ (see Figure~\ref{dQIP_GHZ}), and show the following theorem.
\begin{theorem}
$\mathcal{P}_{GHZ}$ has the following properties:
\begin{itemize}
    \item (completeness): If the prover is honest, the protocol is accepted with probability 1.
    \item (soundness): If the protocol is accepted with probability $\delta$, then the reduced state $\rho$ of the output register $\mathsf{R}_{b_{target}}(u)$ satisfies
    \begin{align*}
        \bra{GHZ}\rho \ket{GHZ} \geq 1- \varepsilon.
    \end{align*}
\end{itemize}
\end{theorem}

\begin{proof} 
We consider the first qubit is the qubit of $\mathsf{leader}$.
Then the test $P_0$ in Figure 1 are described as follows:
\begin{align*}
P_0 = \frac{1}{2}\Big( I + X_1\bigotimes_{i =2}^n Z_i^{\otimes n-1} \Big) 
	&= \frac{1}{2}\Big( I + \underset{z\in \{0,1\}^{n-1}}{\sum}(-1)^{|z|}\big[ 
	\ket{+,z}\bra{+,z} - \ket{-,z}\bra{-,z}
\big] \Big) \\
&=  \underset{z:(-1)^{|z|}=1}{\sum} \ket{+,z}\bra{+,z} + \underset{z:(-1)^{|z|}=-1}{\sum}  \ket{-,z}\bra{-,z}.
\end{align*}
Here $|z|$ is the number of 1's in $z$ and $Z_i$ is the Z gate acting on the $i$-th qubit. We also have
\begin{align*}
P_1 &= \prod_{i = 2}^n \frac{1}{2}\Big( I + Z_1\otimes X_i \bigotimes_{j\neq 1,i} I_{j} \Big) \\
 &= \prod_{i = 2}^n \Big( \ket{0}\bra{0}_1 \otimes  \ket{+}\bra{+}_{i}  \bigotimes_{j\neq 1,i} I_{j} + \ket{1}\bra{1}_1 \otimes \ket{-}\bra{-}_{i}  \bigotimes_{j\neq 1,i} I_{j} \Big) \\
 &= \ket{0,+^{n-1}}\bra{0,+^{n-1}} + \ket{1,-^{n-1}}\bra{1,-^{n-1}}.
\end{align*}
Therefore the conditions (c) and (d) in the verification phase actually checks tests $P_0,P_1$, respectively.
We can analyze the completeness and the soundness as follows:

\noindent
\textbf{Completeness:} The honest prover simply sends $\ket{S_n}^{\otimes N+1}$ in the first turn and broadcasts $b_{test}$, $b_{target}$ of $\mathsf{leader}$ in the third turn. Therefore the protocol is accepted with probability 1. From Lemma~\ref{lem1} the output state is $\ket{GHZ}$.

\noindent
\textbf{Soundness:} Since the protocol simulates $\mathcal{P}_{ZH}$ shown in Figure~\ref{ZH}, by using Theorem~\ref{thm:ZH}, if the protocol is accepted with probability $\delta$, then the reduced state $\rho'$ of regsiter $\mathsf{R}_{b_{target}}$ satisfies
$
    \bra{S_n}\rho' \ket{S_n} \geq 1- \varepsilon.
$
Since the verifier applies $(I\otimes H^{\otimes n-1})$ to $\mathsf{B}_{b_{target}}$, the output state $\rho$ satisfies
\begin{align*}
        \bra{GHZ}\rho \ket{GHZ} \geq 1- \varepsilon.
\end{align*}
\end{proof}

%%%%%%%%%%%%%%%%%%%%%%%%%%%%%%%%%%%%%%%%%%%%
\section{Proof of Theorem~\ref{thm:closeness}}\label{appendix:dqct}
%%%%%%%%%%%%%%%%%%%%%%%%%%%%%%%%%%%%%%%%%%%%

We show the protocol $\mathcal{P}_{\mathsf{DQCT}}$ in Figure~\ref{dSWAP} and analyze the completeness and the soundness.

\begin{lemma}\label{lem:acc_prob_swap}
%Let $\mathcal{H}$ be the Hilbert space corresponds to the state in $\mathsf{R_1,R_2}$. 
Let $\mathsf{R}_1$ and $\mathsf{R}_2$ be the two input registers of $\mathsf{DQCT}_N$.
Assume that  after step~1 of $\mathcal{P}_{\mathsf{DQCT}}$, the state in $(\mathsf{B,R_1,R_2})$ is $\ket{GHZ}\bra{GHZ}_{\mathsf{B}}\otimes \ket{\psi}\bra{\psi}_{\mathsf{R}_1}\otimes \ket{\phi}\bra{\phi}_{\mathsf{R}_2}$, %where $\rho_1$ and $\rho_2$ are two arbitrary density matrices on $\mathcal{H}$.
Then, the acceptance probability of $\mathcal{P}_{\mathsf{DQCT}}$ is at most $\frac{1}{2}+\frac{1}{2}|\braket{\psi|\phi}|^2$.
%$\mathrm{tr}[\Pi U\ket{GHZ}\bra{GHZ}_{\mathsf{B}}\otimes \rho_1\otimes \rho_2 U^{\dagger}]\leq \frac{1}{2}+\frac{1}{2}\mathrm{tr}(\rho_1\rho_2)$, where $\Pi=\ket{0}\bra{0}_{\mathsf{B}}\otimes I_{(\mathsf{R_1,R_2})}$.
\end{lemma}

\begin{proof}
After applied the controlled SWAP gates at step 2, the state is written as
\begin{align*}
    &\frac{1}{\sqrt{2}} \ket{0}_{\mathsf{B}'} \ket{0^n}_{\mathsf{B}} \ket{\psi}_{\mathsf{R}_1}\ket{\phi}_{\mathsf{R}_2}
+ \frac{1}{\sqrt{2}} \ket{1}_{\mathsf{B}'} \ket{1^n}_{\mathsf{B}}\ket{\phi}_{\mathsf{R}_1}\ket{\psi}_{\mathsf{R}_2}.
\end{align*}
After received the register $\mathsf{B}$ in the second turn, the prover performs arbitrary quantum operation, which is followed by $\mathsf{leader}$'s CNOT operation. Note that these procedures do not change the state in $\mathsf{B}'$ since $\mathsf{B}'$ is used as the control qubit.
Thus the entire state can be written  
\begin{align*}
    &\frac{1}{\sqrt{2}} \ket{+}_{\mathsf{B}'} \ket{\xi_0}_{(\mathsf{B,P})} \ket{\psi}_{\mathsf{R}_1}\ket{\phi}_{\mathsf{R}_2}
      + \frac{1}{\sqrt{2}} \ket{-}_{\mathsf{B}'} \ket{\xi_1}_{(\mathsf{B,P})} \ket{\phi}_{\mathsf{R}_1}\ket{\psi}_{\mathsf{R}_2}\\
      &= \frac{1}{2} \ket{0}_{\mathsf{B}'} \Big( 
      	\ket{\xi_0}_{(\mathsf{B,P})} \ket{\psi}_{\mathsf{R}_1}\ket{\phi}_{\mathsf{R}_2} + \ket{\xi_1}_{(\mathsf{B,P})} \ket{\phi}_{\mathsf{R}_1}\ket{\psi}_{\mathsf{R}_2}
	\Big)\\
      &\hspace{4mm} + \frac{1}{2} \ket{1}_{\mathsf{B}'} \Big(
      	\ket{\xi_0}_{(\mathsf{B,P})} \ket{\psi}_{\mathsf{R}_1}\ket{\phi}_{\mathsf{R}_2} - \ket{\xi_1}_{(\mathsf{B,P})} \ket{\phi}_{\mathsf{R}_1}\ket{\psi}_{\mathsf{R}_2}
      \Big)
\end{align*}
for some quantum state $\ket{\xi_0}, \ket{\xi_1}$, where $\mathsf{P}$ is the prover's private space.
Since $\mathsf{leader}$ rejects when the state in $\mathsf{B}'$ is $\ket{1}$, we are interested in the squared value of the amplitude of

\begin{align*}
    \frac{1}{2} \ket{0}_{\mathsf{B}'} \Big( 
      	\ket{\xi_0}_{(\mathsf{B,P})} \ket{\psi}_{\mathsf{R}_1}\ket{\phi}_{\mathsf{R}_2} + \ket{\xi_1}_{(\mathsf{B,P})} \ket{\phi}_{\mathsf{R}_1}\ket{\psi}_{\mathsf{R}_2}
	\Big).
\end{align*}
Since the verifier only accepts when the state in $\mathsf{B}$ is $\ket{0^n}$,
the acceptance probability is maximized when $\ket{\xi_{0}}=\ket{0^n}_{\mathsf{B}}\ket{\alpha_0}_{\mathsf{P}}$ and $\ket{\xi_{1}} = \ket{0^n}_{\mathsf{B}}\ket{\alpha_1}_{\mathsf{P}}$ for some $\ket{\alpha_0}_{\mathsf{P}}$ and $\ket{\alpha_1}_{\mathsf{P}}$. The maximum value is $\frac{1}{2}+\frac{1}{2}|\braket{\psi|\phi}|^2$, which corresponds to $\ket{\alpha_0}=\ket{\alpha_1}$.

\end{proof}

\begin{lemma}\label{lem:}
The protocol $\mathcal{P}_{\mathsf{DQCT}}$ is accepted with probability at most $\frac{1}{2}(1+|\braket{\psi|\phi}|^2) + \sqrt{2\varepsilon}$ where $\varepsilon$ is the parameter appeared in Step 1 of the protocol.
\end{lemma}

\begin{proof}
Assume that $\mathcal{P}_{GHZ}$ in Figure~\ref{dSWAP} is accepted with at least probability $\delta$. %Define a quantity $d= 1-\delta - \Big( \frac{1}{2}+\frac{1}{2}\mathrm{tr}(\rho_1\rho_2) \Big)$. If $d\leq 0$, it means that the required soundness condition holds. Thus we assume that $d>0$. 
From Theorem~\ref{thm:protocol_2}, the reduced state $\rho$ in $\mathsf{B}$ after the step 1 of the protocol satisfies $\bra{GHZ}\rho\ket{GHZ}\geq 1-\varepsilon$ if the verification of the subprotocol $\mathcal{P}_{GHZ}$ is accepted. Let $U$ be the unitary that is applied by the prover and the verifier at steps 2 and 3 of the protocol and $\Pi = \ket{0}\bra{0}_{(\mathsf{B},\mathsf{B}')}\otimes I$. We can see $\mathrm{tr}[\Pi U\ket{GHZ}\bra{GHZ}_{\mathsf{B}}\otimes \ket{\psi}\bra{\psi}\otimes \ket{\phi}\bra{\phi} U^{\dagger}]\leq \frac{1}{2}+\frac{1}{2}|\braket{\psi|\phi}|^2$ by Lemma~\ref{lem:acc_prob_swap}.
Therefore the acceptance probability of $\mathcal{P}_{\mathsf{DQCT}}$ is
\begin{align*}
    %\delta &\cdot \mathrm{tr}[\Pi U\rho\otimes \rho_1\otimes \rho_2 U^{\dagger}] \\
    &\mathrm{tr}[\Pi U\rho\otimes \ket{\psi}\bra{\psi}\otimes \ket{\phi}\bra{\phi} U^{\dagger}] \\
    &\leq  \mathrm{tr}[\Pi U\ket{GHZ}\bra{GHZ}_{\mathsf{B}}\otimes \ket{\psi}\bra{\psi}\otimes \ket{\phi}\bra{\phi} U^{\dagger}] + \mathrm{dist}(\rho,\ket{GHZ}\bra{GHZ})\\
    &\leq \mathrm{tr}[\Pi U\ket{GHZ}\bra{GHZ}_{\mathsf{B}}\otimes \ket{\psi}\bra{\psi}\otimes \ket{\phi}\bra{\phi} U^{\dagger}] + \sqrt{1-|\bra{GHZ}\rho\ket{GHZ}|^2}\\
    &\leq \frac{1}{2}+\frac{1}{2}|\braket{\psi|\phi}|^2 + \sqrt{2\varepsilon}.
\end{align*}

\end{proof}

We are now ready to prove Theorem~\ref{thm:closeness}.

\setcounter{theorem}{3}
\begin{theorem}[restated]
There is a $\mathsf{dQIP}[5](O(1))$ protocol for $\mathsf{DQCT}_N$, where the completeness and the soundness conditions are defined as follows:
\begin{itemize}
    \item \textbf{Completeness:} If $\ket{\psi} = \ket{\phi}$ and the prover is honest, the protocol is accepted with probability~1.
    %\item  For any $\ket{\psi}$ and $\ket{\phi}$, the protocol is accepted with probability at most $\frac{1}{2}+\frac{1}{2}|\braket{\psi|\phi}|^2 + \varepsilon$ for any small constant $\varepsilon > 0$.
    \item \textbf{Soundness:} If the protocol is accepted with probability $1-1/z$, $\mathrm{dist}(\ket{\psi},\ket{\phi}) \leq \sqrt{2/z} + \varepsilon$ for any small constant $\varepsilon > 0$.
\end{itemize}
\end{theorem}
\setcounter{theorem}{26}
\begin{proof}
{\bf Completeness:} Assume that the contents of two input registers $\mathsf{R}_1$ and $\mathsf{R}_2$ are identical. The prover simulates honest operations for the run of ${\mathcal P}_{GHZ}$ in step 1.
After received the register $\mathsf{B}$ in Turn 4, the prover applies $n-1$ CNOT gates where the control qubit is $\mathsf{B}(\mathsf{leader})$ and the target qubit is the other part of $\mathsf{B}$. Therefore the state in $\mathsf{B}',\mathsf{B}$ becomes
\begin{align*}
\frac{1}{\sqrt{2}}
\Big(
\ket{0}_{\mathsf{B}'}\ket{0}_{\mathsf{B}(\mathsf{leader})} 
+
\ket{1}_{\mathsf{B}'}\ket{1}_{\mathsf{B}(\mathsf{leader})}  
\Big) \otimes \ket{0^{n-1}}_{\mathsf{B}\backslash \mathsf{B}(\mathsf{leader})}.
\end{align*}
 It can be checked easily that after the local operations of $\mathsf{leader}$ at the verification phase, the resulting state in $(\mathsf{B',R_1,R_2})$ is identical to the state before the measurement of the SWAP test. Thus the protocol is accepted with probability~1.
 
 \noindent
 {\bf Soundness:} From Lemma~\ref{lem:acc_prob_swap}, we have $\frac{1}{2}+ \frac{1}{2}|\braket{\psi|\phi}|^2 + \sqrt{2\varepsilon} \geq 1-\frac{1}{z}$.
 Then we have
\begin{align*}
	1-\frac{2}{z}-\sqrt{8\varepsilon} \leq  |\braket{\psi | \phi}|^2 = 
 F(\ket{\psi}\bra{\psi},\ket{\phi}\bra{\phi})^2.
\end{align*}
%from the concavity of the fidelity, we have
By Lemma~\ref{lem:inequalities}, 
\begin{align*}
\mathrm{dist}(\ket{\psi},\ket{\phi})\leq \sqrt{1-F(\ket{\psi}\bra{\psi},\ket{\phi}\bra{\phi})^2}\leq \sqrt{1-(1-2/z-\sqrt{8\varepsilon})} \leq \sqrt{\frac{2}{z} + \sqrt{8\varepsilon}},
\end{align*}
which is bounded by $\sqrt{2/z} + \varepsilon'$ for arbitrary small $ \varepsilon'>0$ by taking sufficiently small $\varepsilon>0$.
Thus the proof is completed.
\end{proof}

\begin{figure}[htbp]
\begin{mdframed}
\textbf{INPUT:} each node $u$ has two registers $\mathsf{R}_1(u)$ and $\mathsf{R}_2(u)$ that are not entangled each other.
\begin{enumerate}
    \item \textbf{Turn 1-3:} Let $0<\varepsilon<1$, $0<\delta<1$ be two positive numbers and $N\in \Theta(\frac{1}{\varepsilon}\log\frac{1}{\delta})$ be an integer apprearing in Theorem~\ref{thm:ZH}. The prover and the verifier perform the protocol $\mathcal{P}_{GZH}$ in the Figure 1 with these parameters, except the verification phase. %twice in parallel.
     \item \textbf{Turn 1-3:} The prover and the verifier construct a spanning tree rooted at $\mathsf{leader}$.
     \item Let $\mathsf{B}$ be the output register of $\mathcal{P}_{GHZ}$, which (ideally) contains $\ket{GHZ}$. $\mathsf{leader}$ prepares 1-qubit register $\mathsf{B'}$ and applies the CNOT gate with the target qubit $\mathsf{B}'$ and the control qubit $\mathsf{B(\mathsf{leader})}$.
 Each node $u$ performs controlled-SWAP gate to $\mathsf{R}_1(u)$ and $\mathsf{R}_2(u)$ using controlled register $\mathsf{B}(u)$. 
  \item \textbf{Turn 4:}  Each node sends $\mathsf{B}(u)$ to the prover.
    \item \textbf{Turn 5:} The prover sends $\mathsf{B}(u)$ to each node $u$.
    \item (Verification Phase) 
    $\mathsf{leader}$ applies the CNOT gate with the control qubit $\mathsf{B}'$ and the target qubit $\mathsf{B(\mathsf{leader})}$, then, applies the Hadamard gate to $\mathsf{B}'$. Each node $u$ performs the verification of the constructed spanning tree, and rejects if the verification fails.
    Each node $u$ runs the verification phase of $\mathcal{P}_{GHZ}$ and measures $\mathsf{B}(u)$ in the computational basis. Each node $u$ accepts iff $\mathcal{P}_{GHZ}$ is accepted and the content of $\mathsf{B}(u)$ is $\ket{0}$.
\end{enumerate}
\end{mdframed}
\caption{$\mathsf{dQIP}$ protocol $\mathcal{P}_{\mathsf{DQCT}}$}
\label{dSWAP}
\end{figure}

%%%%%%%%%%%%%%%%%%%%%%%%%%%%%%%%%%%%%%%%%
\section{Perfect completeness}\label{section:perfect}
%%%%%%%%%%%%%%%%%%%%%%%%%%%%%%%%%%%%%%%%%%
In this appendix we consider the $\mathsf{dQIP}c$ model, and
show how to transform a protocol with two-sided bounded error into a protocol with perfect completeness. 
%one-sided error. 
The number of turn of the transformed protocol increases by four turns and the size of message registers remains unchanged. This is shown by implementing the result of \cite{kitaev2000parallelization} in a distributed manner. The main difference from the centralized setting is the rejection condition. \par
In the case of the distributed setting, assume each node outputs 0 if it accepts otherwise outputs 1 in the case of the distributed setting.
Then the protocol is rejected iff the output is not all-zero. The method of Kitaev and Watrous in the case of the centralized setting includes the operation that an additional register is prepared by the verifier, and it is incremented iff the private register is in the rejection state (i.e., $\ket{1}$). In the distributed setting, the prover determines the leader node, and the leader performs this operation. However, as mentioned above, in the case of distributed setting, the protocol rejects even if the leader is in the acceptance state but some other node is in the rejection state. Hence the leader needs the help of the prover to confirm that the protocol is in the rejected state. These operations require four additional turns but do not change the size of the message register.

Let $\mathcal{L}$ be a language that has a $k$-turn $\mathsf{dQIP}c$ system with completeness $c$ and soundness $s$ where $c-s>\delta$ for some constant $\delta > 0$. We show a distributed implementation of Kitaev and Watrous in Figure~\ref{fig:protocol_perfect}. Note that this protocol works for $\mathsf{dQIP}c$ but does not work for $\mathsf{dQIP}$ (and $\mathsf{dQIP}^{sh}$) since nodes need to communicate with each other in the middle of the protocol in order to check if its private register is in the acceptance state. The following theorem can be shown easily by adapting the proof of~\cite{kitaev2000parallelization} to the distributed setting.
\begin{theorem}\label{theorem:completeness}
Any $\mathsf{dQIP}c[k](f(n))$ protocol with completeness $c$ and soundness $s$ where $c-s>\delta$ for some constant $\delta > 0$ that uses $(g(n)+f(n)\mathrm{deg}(u))$-qubit private register at node $u$ can be transformed to a $\mathsf{dQIP}c[k+4](f(n)\Delta +g(n))$ protocol that uses $(g(n)+f(n)\mathrm{deg}(u))$-qubit private register at node $u$, with perfect completeness and soundness $1-\delta^2$ where $\Delta$ is the maximum degree of the network.
\end{theorem}
%Using parallel repetition for quantum protocols~\cite{gutoski2010quantum}, 
Then we can use the $\mathsf{dQIP}c$ variant of Theorem \ref{theorem_5_turn} and the parallel repetition of~\cite{gutoski2010quantum} repeatedly to reduce the number of turns to $5$.
\begin{corollary}\label{corollary:completeness}
Any $\mathsf{dQIP}c[k](f(n))$ protocol that uses $(g(n)+f(n)\mathrm{deg}(u))$-qubit private register at node $u$ can be transformed to a $\mathsf{dQIP}c[5](f(n)\Delta + g(n))$ protocol with perfect completeness where $\Delta$ is the maximum degree of the network.
\end{corollary}

\begin{figure}[htbp]
\begin{mdframed}
\begin{enumerate}
	\item Run the original protocol except outputting accept or reject. Construct a spanning tree with the root $\ell$.
	\item Each node sends to the prover the first qubit of its private register which corresponds to the output.
	\item The prover sends one-qubit registers to all nodes. The network checks if all of qubits provided by the prover are the same. If not, the network rejects the protocol.
	\item Let $\ell$ be the leader node which is also the root of a spanning tree. The leader $\ell$ prepares one-qubit registers $\mathsf{B}$ and $\mathsf{B}'$ in the state $\ket{0}\ket{0}$ and increments both $\mathsf{B}$ and $\mathsf{B}'$ iff the first qubit of its private register is $\ket{1}$.
	\item Each node $u$ sends its private register $\mathsf{V}_u$ and $\mathsf{W}_{u,v}$ for all $(u,v)\in E$ to the prover. The leader sends $\mathsf{B}$ to the prover.
	\item The prover sends $\mathsf{B}'$ to $\ell$, and $\ell$ subtracts $\mathsf{B}$ from $\mathsf{B}'$ (i.e., flips $\mathsf{B}'$ iff the content of $\mathsf{B}$ is $\ket{1}$). The leader $\ell$ applies $T_c$ to $\mathsf{B}$ where $T_c$ is given by
	\begin{align*}
	    T_c(\ket{0})= \sqrt{c}\ket{0} - \sqrt{1-c}\ket{1}\\
	    T_c(\ket{1})= \sqrt{1-c}\ket{0} + \sqrt{c}\ket{1}.
	\end{align*}
	Then the leader measures $\mathsf{B}$ and accepts iff the outcome is $\ket{0}$.
\end{enumerate}
\end{mdframed}
\caption{$(k + 4)$-turn $\mathsf{dQIP}c$ protocol with perfect completeness.}
\label{fig:protocol_perfect}
\end{figure}

\end{appendix}

\end{document}